\documentclass[prd,aps,amsfonts,eqsecnum,superscriptaddress,nofootinbib,notitlepage,longbibliography,10pt]{revtex4-2}
\usepackage{booktabs}

\usepackage[a4paper, left=1.8cm, right=1.8cm, top=2cm, bottom=2cm]{geometry}
\usepackage{setspace}
\usepackage[english]{babel}
\addto\extrasenglish{%
	\def\appendixautorefname{Appendix}%
}
\usepackage{appendix}
\usepackage[utf8]{inputenc}
\usepackage[T1]{fontenc}
\usepackage{lmodern}
\usepackage{amsmath,amsthm,mathdots,amssymb,bm,bbm,mathrsfs,mathtools}
\usepackage[final]{graphicx} 
\usepackage{subcaption}
\usepackage{tikz,pgfplots}\pgfplotsset{compat=1.16}

\usepackage{thm-restate}
\newtheorem{thm}{Theorem}[section]
\newtheorem{cor}{Corollary}[section]
\newtheorem{lem}{Lemma}[section]

\newtheorem{alg}{Algorithm}[section]
\newtheorem{rmk}{Remark}[section]

\newtheorem{cond}{Condition}[section]
\usepackage{comment}

\usepackage[colorlinks=true, urlcolor=violet, linkcolor=blue, citecolor=red, hyperindex=true, linktocpage=true, pagebackref=true, draft=false]{hyperref}
\usepackage{mleftright}\mleftright 

\usepackage{enumitem}
\mathtoolsset{showonlyrefs=true}
\usepackage{microtype,color,graphicx} 
\usepackage{tikz-cd}
\usepackage{xcolor}
\usepackage{multirow}
\usepackage{tablefootnote}

\renewcommand{\vec}{\bm}

\newcommand{\CC}{\mathcal{C}}

\newcommand{\CO}{\mathcal{O}}

\newcommand{\BR}{\mathbb{R}}

\newcommand{\vA}{\bm{A}}

\newcommand{\vB}{\bm{B}}

\newcommand{\vC}{\bm{C}}
\newcommand{\vD}{\bm{D}}
\newcommand{\vE}{\bm{E}}
\newcommand{\vF}{\bm{F}}
\newcommand{\vf}{\bm{f}}
\newcommand{\vG}{\bm{G}}
\newcommand{\vg}{\bm{g}}
\newcommand{\vH}{\bm{H}}
\newcommand{\vh}{\bm{h}}

\newcommand{\vI}{\bm{I}}

\newcommand{\vK}{\bm{K}}
\newcommand{\vk}{\bm{k}}

\newcommand{\vO}{\bm{O}}
\newcommand{\vP}{\bm{P}}

\newcommand{\vQ}{\bm{Q}}

\newcommand{\vT}{\bm{T}}

\newcommand{\vU}{\bm{U}}
\newcommand{\vV}{\bm{V}}
\newcommand{\vW}{\bm{W}}
\newcommand{\vX}{\bm{X}}
\newcommand{\vY }{\bm{Y }}
\newcommand{\vZ}{\bm{Z}}

\newcommand{\vrho}{\bm{ \rho}}
\newcommand{\vtau}{\bm{ \tau}}

\renewcommand{\L}{\left}
\newcommand{\R}{\right}

\newcommand{\vertiii}[1]{{\left\vert\kern-0.25ex\left\vert\kern-0.25ex\left\vert #1 \right\vert\kern-0.25ex\right\vert\kern-0.25ex\right\vert}}
\newcommand{\norm}[1]{\Vert {#1} \Vert}

\newcommand{\labs}[1]{\left\vert {#1} \right\vert}
\newcommand{\lnorm}[1]{\left\Vert {#1} \right\Vert}
\newcommand{\e}{\mathrm{e}}

\newcommand{\ri}{\mathrm{i}}
\newcommand{\rd}{\mathrm{d}}
\newcommand*{\tr}{\mathrm{Tr}}
\newcommand*{\poly}{\mathrm{Poly}}

\newcommand*{\Supp}{\mathrm{Supp}}
\newcommand{\indicator}{\mathbbm{1}}

\DeclarePairedDelimiterX{\braket}[1]{\langle}{\rangle}{#1}
\DeclarePairedDelimiterX{\Braket}[1]{\bigg\langle}{\bigg\rangle}{#1}
\DeclarePairedDelimiterX\ketbra[2]{| }{|}{#1 \delimsize\rangle\!\delimsize\langle #2}	
\DeclarePairedDelimiterX\dotp[2]{\langle}{\rangle}{#1, #2}

\newcommand{\bigOt}[1]{\widetilde{\mathcal{O}}\left( #1 \right)}

\DeclareMathAlphabet{\dutchcal}{U}{dutchcal}{m}{n}
\SetMathAlphabet{\dutchcal}{bold}{U}{dutchcal}{b}{n}
\DeclareMathAlphabet{\dutchbcal} {U}{dutchcal}{b}{n}

\makeatletter
\DeclareRobustCommand*{\pmzerodot}{%
	\nfss@text{%
		\sbox0{$\vcenter{}$}
		\sbox2{0}%
		\sbox4{0\/}%
		\ooalign{%
			0\cr
			\hidewidth
			\kern\dimexpr\wd4-\wd2\relax 
			\raise\dimexpr(\ht2-\dp2)/2-\ht0\relax\hbox{%
				\if b\expandafter\@car\f@series\@nil\relax
				\mathversion{bold}%
				\fi
				$\cdot\m@th$%
			}%
			\hidewidth
			\cr
			\vphantom{0}
		}%
	}%
}

\usepackage{ifdraft}

\usetikzlibrary{quantikz}
\makeatletter
\def\l@subsubsection#1#2{}
\makeatother
\begin{document}
\def\appendixautorefname{Appendix}
\renewcommand{\appendixautorefname}{Appendix}
\renewcommand{\chapterautorefname}{Chapter}
\renewcommand{\sectionautorefname}{Section}
\renewcommand{\subsectionautorefname}{Section}
\renewcommand{\subsubsectionautorefname}{Section}	

\title{Learning quantum Gibbs states locally and efficiently}

\author{Chi-Fang Chen}
\email{achifchen@gmail.com}
 \affiliation{University of California, Berkeley, CA, USA }
\affiliation{Massachusetts Institute of Technology, Cambridge, USA}
\author{Anurag Anshu}
\email{anuraganshu@fas.harvard.edu}
\author{Quynh T. Nguyen}
\email{qnguyen@g.harvard.edu}
\affiliation{School of Engineering and Applied Sciences, Harvard University}

\begin{abstract}
    Learning the Hamiltonian underlying a quantum many-body system in thermal equilibrium is a fundamental task in quantum learning theory and experimental sciences. To learn the Gibbs state of local Hamiltonians at any inverse temperature $\beta$, the state-of-the-art provable algorithms fall short of the optimal sample and computational complexity, in sharp contrast with the locality and simplicity in the classical cases. In this work, we present a learning algorithm that learns each local term of a $n$-qubit $D$-dimensional Hamiltonian to an additive error $\epsilon$ with sample complexity $\tilde{O}\left(\frac{e^{\mathrm{poly}(\beta)}}{\beta^2\epsilon^2}\right)\log(n)$. The protocol uses parallelizable local quantum measurements that act within bounded regions of the lattice and near-linear-time classical post-processing. Thus, our complexity is near optimal with respect to $n,\epsilon$ and is polynomially tight with respect to $\beta$. We also give a learning algorithm for Hamiltonians with bounded interaction degree with sample and time complexities of similar scaling on $n$ but worse on $\beta, \epsilon$. At the heart of our algorithm is the interplay between locality, the Kubo-Martin-Schwinger condition, and the operator Fourier transform at arbitrary temperatures. 
\end{abstract}

\maketitle
\tableofcontents
\section{Introduction}
Identifying and testing the physical laws governing the interactions between particles is a fundamental quest of quantum many-body physics. Even though all quantum phenomena in nature can be, in principle, reduced to the standard model for elementary particles, an effective description at relevant scales is needed to make meaningful predictions. In the language of quantum mechanics - and assuming the locality of nature - local Hamiltonians provide a minimal effective framework governing the relevant degree of freedom. Recovering the local Hamiltonians for strongly interacting systems has often led to major leaps in physics \cite{BCS57, Laughlin83} and has paved the way for technological advances such as high-temperature superconductors and quantum information processing platforms. 

With the growing confluence between computer science and the physical sciences, a natural question arises. Can the search for the underlying local Hamiltonian be automated, especially when quantum metrology protocols become more controllable and robust? The recent pursuit for efficient Hamiltonian learning precisely captures this goal and sits at the forefront of research in quantum learning theory \cite{EyalAL19, anshu2020sample, RF24, HKT22, HTFS23, BLMT24}, quantum computing, and experimental physics \cite{Kokail2021,OKKKKZ25}.

This work focuses on learning the underlying Hamiltonian of a quantum system in thermal equilibrium. In particular, we consider a minimal experimental setting without any access to the real-time \textit{dynamics}, but only \textit{static} observables. Formally, we assume the thermal equilibrium is modeled by the \textit{Gibbs} state at inverse temperature $\beta$
\begin{align}
\vrho_{\beta}(\vH)=\frac{e^{-\beta \vH}}{\tr(e^{-\beta \vH})}    
\end{align}
of an unknown local Hamiltonian $\vH$. The goal, then, is to output a Hamiltonian $\vH'$ that is close enough to $\vH$, using as few independent samples from $\vrho_{\beta}(\vH)$ and as simple observables as possible. Practical, efficient, and simple Hamiltonian learning techniques open doors to other applications in experimental settings, notably verifying quantum devices where the Hamiltonian serves as the hidden parameter. 

The state-of-the-art algorithms fall short of the optimal sample and time complexity. The work \cite{anshu2020sample} showed that information-theoretically, Gibbs states are uniquely determined by the collection of all few-qubit observables. However, computationally, recovering the Hamiltonian from Gibbs marginals, even for classical Hamiltonians, is generally intractable (NP-hard) \cite{montanari2015}. The recent impressive work \cite{BLMT24} achieves polynomial sample and computational complexity at arbitrary constant temperatures. Still, the measurement involves far-apart qubits, the polynomial exponent deteriorates as the temperature lowers, and the sample complexity is exponentially far from the optimal in relevant regimes (e.g., when learning each Hamiltonian term to a constant precision). An algorithm achieving optimal sample and time complexity was known at high temperatures~\cite{HKT22}. However, high temperature is a strong assumption and we may not always have the experimental capacity to change the effective temperature of the sample (such as for spin glasses \cite{Panchenko2012}, where raising temperature may be difficult, or for entanglement Hamiltonians \cite{Kokail2021}, where the effective temperature does not correspond to a physical temperature). Other heuristic algorithms have low sample complexity and time complexity but do not have rigorous guarantees \cite{EyalAL19, LBAALA23}. 

The central conceptual bottleneck in achieving an optimal algorithm at arbitrary temperatures is the unsettled role of locality in Hamiltonian learning: 
\begin{center}
\textit{Should a local Hamiltonian term be uniquely identified by the neighbouring marginal of the Gibbs state?}
\end{center}
This question of \textit{local sufficient statistics} is a strengthening of~\cite{anshu2020sample}, which did not rule out the possibility that \textit{all} marginals would be needed to learn a given local term. For a given local term, local statistics are only known to be sufficient in the high temperature case \cite{HKT22} and in the commuting case. More broadly speaking, the challenge of devising a truly local Gibbs learning algorithm is compounded by the lack of a structural understanding of multipartite entanglement in quantum Gibbs states beyond one dimension. Especially at low temperatures, the quantum system may undergo thermal phase transitions and exhibit long-range quantum and classical correlations.

In this work, we answer in the affirmative and devise a simple protocol that learns the individual terms by individual localized measurements. We begin with the most intuitive algorithm, which works for Hamiltonians on interaction graphs with bounded degree (such as an expander; see \autoref{sec:Ham} for a formal definition). For the following~\autoref{thm:graphthm} and ~\autoref{thm:latticethm} and throughout the paper, $\CO(\cdot)$ denotes an asymptotic upperbound when suppressing the geometric parameters (degree $d$, dimension $D$, and locality $q$) of the Hamiltonian, and $\poly(\cdot)$ denotes a polynomial depending only on $q,d,D$. 
\begin{thm}[Learning each local term locally]
\label{thm:graphthm}
Consider a target local Hamiltonian $\vH = \sum_{\gamma\in \Gamma} \vh_{\gamma}$ with a constant interaction degree (as in~\autoref{sec:Ham}) with the promise that
\begin{align}
    \vh_{\gamma} = h_{\gamma} \vP_{\gamma}\quad \text{for unknown coefficients}\quad h_{\gamma}\in [-1,1] 
\end{align}
and for known Pauli operator $\vP_{\gamma}$ of constant weight. Suppose we have access to its Gibbs state $\vrho_{\beta}$ at a known inverse temperature $\beta>0$. Then, there is a protocol that learns each coefficient $h_{\gamma}$ up to an additive error $\epsilon$ with success probability at least $1-\delta$ with
\begin{center} 
sample complexity $\quad\CO\L(  \log (n/\delta)\cdot 2^{\mathrm{poly}(1/\beta\epsilon) 2^{\CO(\beta^4)} }  \R)$\\
and time complexity $\quad\CO\L(n\log(n/\delta) \cdot 2^{\mathrm{poly}(1/\beta\epsilon) 2^{\CO(\beta^4)}}\R)$. \end{center}
\end{thm}
See~\autoref{sec:local_anygraph} for the explicit algorithm and the proof. The claim is that each term can be individually learned given neighboring marginals of a system-size independent radius $\CO(\beta^4 + (\beta+1)\log(1/\epsilon))$ in the graph distance. The sample complexity simply follows from performing parallelizable local measurements, achieving the $\log(n)$ scaling when $\epsilon, \beta = \Theta(1)$. Here, the poor scaling with $\epsilon$ and $\beta$ roots from searching over all possible Hamiltonians within the radius, which can be large on expander graphs.  

In physical settings, Hamiltonians are often defined on $D$-dimensional lattices (particularly we assume $q=\CO(1),d=\CO(1)$; see~\autoref{sec:Ham} for the precise definition). Further exploiting geometric locality, we give a provably efficient local learning algorithm and settle the Hamiltonian learning problem on $D$-dimensional lattices at any temperature (see~\autoref{sec:high_precision_algorithm} for the algorithm and the analysis).

\begin{thm}[Learning $D$-dimensional Hamiltonians]
\label{thm:latticethm}
Consider a $D$-dimensional Hamiltonian $\vH = \sum_{\gamma\in \Gamma} \vh_{\gamma}$ (as in~\autoref{sec:Ham}) with the promise that
\begin{align}
    \vh_{\gamma} = h_{\gamma} \vP_{\gamma}\quad \text{for unknown coefficients}\quad h_{\gamma}\in [-1,1] 
\end{align}
for known Pauli operators $\vP_{\gamma}$ of constant weight. Suppose we have access to samples of its Gibbs states at a known inverse temperature $\beta>0$. Then, there is a protocol that learns each coefficient $h_{\gamma}$ up to an additive error $\epsilon$ with success probability at least $1-\delta$ with
\begin{center}
sample complexity
$\quad\CO\L(\log (n/\delta) \cdot \frac{e^{\poly(\beta)}}{\beta^2\varepsilon^2} \poly\log(1/\varepsilon) \R)$\\
and time complexity
$\quad\CO\L(n\log (n/\delta)\cdot \frac{e^{\poly(\beta)}}{\beta^2\varepsilon^2} \poly\log( 1/\varepsilon)\R).$ \end{center}

\end{thm}
Due to the lower bound\footnote{In our paper, $\epsilon$ captures the $\ell_\infty$-learning of the Hamiltonian terms. The case of $\ell_2$-learning can be obtained by setting $\epsilon =\CO(\epsilon_2/\sqrt{n})$.} of \cite[Theorem 1.2]{HKT22} 
\begin{align}
\Omega\L(\frac{e^{\beta}}{\beta^2\epsilon^2}\log\frac{n}{\delta} \R),    
\end{align}
our achieved complexity is optimal in the number of qubits $n$ and (nearly) optimal in the precision $\epsilon$. The algorithm works at all temperatures, and is polynomially tight in the $\beta$ dependence. The efficiency of this protocol stems from an iterative procedure that, in each sweep, takes the current guess as input, performs some measurement, and proposes a better guess with doubled precision. The iterative protocol is not manifestly local as one collectively updates multiple coefficients. Still, if we track the flow of information throughout the iterations, the coefficient in an individual term is still essentially determined by measurement data within a $\poly(\log(1/\epsilon),\beta)$ radius. 

\subsection{The protocol and key ideas}

Our local approach to Hamiltonian learning is made possible by recent developments in our understanding of the dynamical origin of Gibbs states - quantum Gibbs samplers~\cite{temme2011quantum,yung2012quantum, Shtanko2021AlgorithmsforGibbs, chen2021fast, rall2023thermal, wocjan2023szegedy,chen2023quantum,chen2023efficient, gilyen2024quantum, jiang2024quantum,ding2024single,ding2024efficient}. While our algorithm does not explicitly implement a Lindbladian dynamics, we rely heavily on the fundamental analytic toolkit introduced in~\cite{chen2023quantum, chen2023efficient}, which provided a local approach to quantum detailed balance.

\subsubsection{Local sufficient statistics from the KMS condition}
\noindent {\bf Kubo–Martin–Schwinger (KMS) condition:} Quantum states at thermal equilibrium satisfy a kind of microscopic reversibility for any observables. The KMS condition in quantum statistical mechanics is an identity for thermal two-point functions (Green functions) at an inverse temperature $\beta$
 \begin{align}
\tr(\vO\vP_{\vH}(t)\vrho) = \tr(\vP_{\vH}(t+i \beta)\vO \vrho) \quad \text{for every $\vP,\vO$ and $t\in \BR$}, \quad \text{where}\quad \vP_{\vH}(z):=e^{i \vH z}\vP e^{-i \vH z},\label{eq:KMS} 
\end{align}
where we have abbreviated $\vrho:=\vrho_{\beta}(\vH)$ for convenience. In fact, the KMS condition provides a \textit{unique} definition of the Gibbs state in terms of the correlation functions. To see this, consider a state $\vrho$ satisfying the KMS condition for a test Hamiltonian $\vH'$. Then, we may drop the quantifier over variable $\vO$ and denote $\vrho'$ the Gibbs state for $\vH'$ to deduce that 
\begin{align}
\eqref{eq:KMS} &\iff\quad \vP_{\vH'}(t)\vrho = \vrho\vrho'^{-1
}\vP_{\vH'}(t)\vrho' \quad \text{for all }\vP \text{ and } t\in \BR,\\  
&\iff \quad\quad \vrho\vrho'^{-1}\propto \vI,\\
&\iff \quad\quad\quad \beta\vH=\beta\vH' + c \vI. 
 \label{eq:detailedbalintro}  
\end{align} 
The second line uses the fact that the Gibbs state of a bounded Hamiltonian is invertible and that an operator that commutes with all matrices must be proportional to the identity. The third line uses the uniqueness of matrix logarithm for full rank PSD inputs and uses $c\vI$ to account for normalization of Gibbs state. Here, the real-time parameter $t$ does not play a role in the above argument, and indeed, the argument in~\cite{BLMT24} seems to only require the $t=0$ part of the KMS condition. One may wonder whether the KMS condition for local $\vO,\vP$ would lead to the desired local sufficient statistics for identifying the Hamiltonian. 
\vspace{0.1in}

\noindent {\bf The identifiability equation:} 
Of course, the above exact argument~\eqref{eq:detailedbalintro} is fragile, and any physical quantity can only be measured approximately up to statistical errors. The second key idea is to formulate a robust version via the following \textit{identifiability equation} (see~\autoref{lem:identify}), defined for the ground truth $\vH$ and the test Hamiltonian $\vH'$, and every local operator pair $\vA$ and $\vO$
\begin{align}
\frac{\beta}{2} \langle \vO,[\vA,\vH-\vH']\rangle_{\vrho} = \frac{1}{\sqrt{2\pi}}\int_{-\infty}^{\infty} \tr\L[\vO^{\dagger}_{\vH}(t)\L(  \sqrt{\vrho'}\vA_{\vH'}(t)\sqrt{\vrho^{'-1}}\vrho - \vrho\sqrt{\vrho^{'-1}}\vA_{\vH'}(t)\sqrt{\vrho^{'}} \R)\R]g_{\beta}(t) \rd t
\label{eq:identify_intro}
\end{align}
where $\langle \vX, \vY\rangle_{\vrho}:=\tr[\vX^{\dagger}\sqrt{\vrho}\vY\sqrt{\vrho}]$ is the KMS inner product\footnote{Technically, the KMS inner product seems to play a special role. We were not able to change the KMS inner product in Equation~\eqref{eq:identify_intro} as well as the particular powers of $\vrho'$ appearing on the RHS.} and $g_\beta(t)$ is a rapidly decaying function. Setting $\vO = [\vA, \vH-\vH']$, the LHS is a positive quantity that vanishes if and only if $[\vA,\vH]=[\vA,\vH'],$ thus
\begin{align}
    \vH = \vH'\quad\iff\quad \sqrt{\vrho'}\vA_{\vH'}(t)\sqrt{\vrho'^{-1}}\vrho &= \vrho\sqrt{\vrho'^{-1}}\vA_{\vH'}(t)\sqrt{\vrho'}\quad \text{ for all single-site Pauli }\vA \text{ and } t \in \BR.
\end{align}
This identifiability equation is reminiscent of the KMS condition when $\vP := \sqrt{\vrho'}\vA\sqrt{\vrho'^{-1}}$, but is moreover robust and local: whenever the RHS is reported to be small in local measurements, $\vH$ and $\vH'$ must agree when taking local commutators with $\vA$. Curiously, the identifiability Equation~\eqref{eq:identify_intro} involves time dynamics of two distinct Hamiltonians at the same time, which allows us to filter out the linear-order in 
$\vH, \vH'$ from Equation \ref{eq:detailedbalintro}. This requires a non-traditional decomposition into Bohr frequencies of a Hamiltonian pair (see~\autoref{sec:doubleBohr}). 

Two glaring issues remain:
\begin{enumerate}
    \item The imaginary-time evolved operator $\sqrt{\vrho'^{-1}}\vA \sqrt{\vrho'}$ is well-known to be a nasty operator in general. Essentially, in more than one dimension, a local operator $\vA$ may have non-negligible amplitudes-of the order of $e^{-c\nu}$ - that substantially change the energy $\nu \gg 1$. Thus, there may be a constant $\beta$ for which the norm $\norm{\sqrt{\vrho'^{-1}}\vA \sqrt{\vrho'}}$ diverges exponentially with the system size. In particular, the off-diagonal terms of $\sqrt{\vrho'^{-1}}\vA \sqrt{\vrho'}$ in the eigenbasis of $\vH'$ blow-up exponentially (see~\autoref{fig:regularize_imaginary}), and hence, there are no local approximations to this operator \cite{Bouch15}. 
    \item The operator $\vO_{\vH}(t)$ in the RHS of Equation \eqref{eq:identify_intro} actually depends on the unknown Hamiltonian $\vH$; thus, it is a priori unclear how to measure the RHS directly in an experiment.
\end{enumerate}

\begin{figure}[t]
\includegraphics[width=0.7\textwidth]{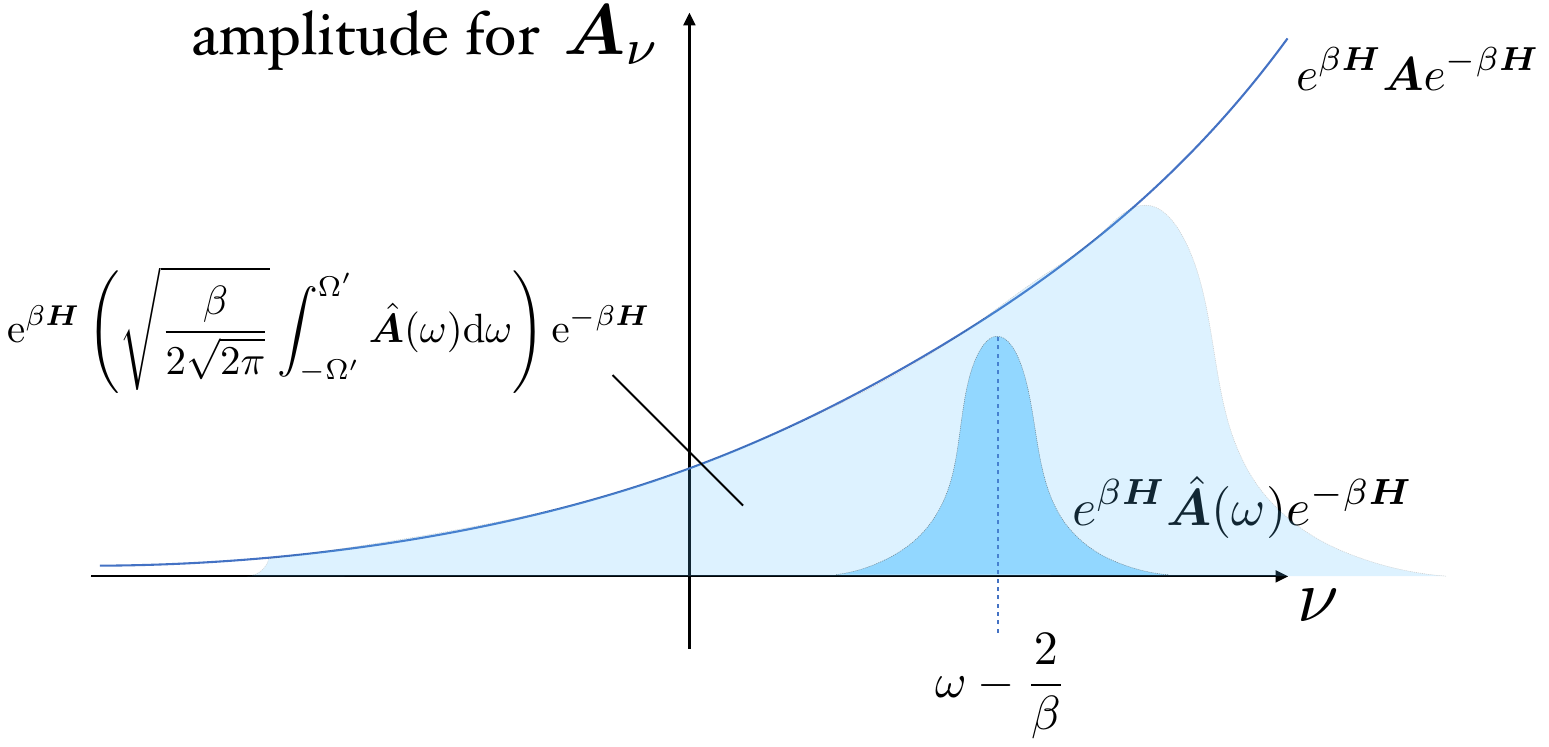}\caption{ 
The imaginary time conjugation $e^{\beta \vH}\vA e^{-\beta \vH} = \sum_{\nu} \e^{\beta \nu}\vA_{\nu}$ assigns exponential weight $e^{\beta \nu}$ according to the Bohr frequency $\nu.$ This exponential growth typically causes trouble in controlling the norm of imaginary time conjugation for a large constant $\beta$, except for one-dimensional spin chains. Remarkably, rewriting $\vA$ in terms of operator Fourier transforms $\hat{\vA}(\omega)$ (setting $\sigma = 1/\beta$ for simplicity) has a convenient regularizing effect, allowing us to separately address the lower frequency parts, which remain controlled under imaginary time conjugation,
and higher frequency parts, which we truncate~\cite{chen2025Markov}. In particular, the Gaussian profile interacts nicely with exponentials, leading to convenient calculation and norm bounds. 
}\label{fig:regularize_imaginary}
\end{figure}

\vspace{0.1in}

\noindent {\bf Solution to Issue 1:} The insight from recent construction of quantum Gibbs samplers~\cite{chen2023quantum,chen2023efficient} is that quantum detailed balance can be imposed locally, by considering the operator Fourier transform for frequency $\omega \in \BR$, which is localized both in frequency and time domain
\begin{align}
\hat{\vA}_{\vH'}(\omega)=  \frac{1}{\sqrt{2\pi}}\int_{-\infty}^{\infty}  \vA_{\vH'}(t)  \e^{-\ri \omega t} f(t)\rd t 
\end{align} 
with the Gaussian weight $f(t) \propto e^{-t^2/\beta^2}$ with time uncertainly $\beta$. An elegant property of the operator Fourier transform is that it behaves nicely under imaginary time evolution (see~\autoref{fig:regularize_imaginary}):
$$\sqrt{\vrho'}\hat{\vA}_{\vH'}(\omega) \sqrt{\vrho'^{-1}} =  \hat{\vA}_{\vH'}(\omega+4/\beta) e^{\beta \omega/2 + 1}.$$
Further, the operator Fourier transform is a linear combination of real-time dynamics $\vA(t)$, giving a (quasi)-local characterization of the KMS condition~\eqref{eq:detailedbalintro} in the frequency space (See also~\autoref{sec:OFT})
\begin{equation}
\label{eq:detailedbalanceloc}
\vH=\vH' \Leftrightarrow \hat{\vA}_{\vH'}(\omega-4/\beta) \vrho = \vrho\hat{\vA}_{\vH'}(\omega+4/\beta) e^{\beta \omega} \quad \text{for each}\quad \omega.
\end{equation}
Crucially, both $\hat{\vA}(\omega-4/\beta)$ and $\hat{\vA}(\omega+4/\beta)$ are now quasi-local observables.

To weave Equation \eqref{eq:detailedbalanceloc} into Equation \eqref{eq:identify_intro}, we follows the very recent work~\cite{chen2025Markov} (see~\autoref{sec:regularize}). For any operator $\vA,$ we can manually split it into the low-frequency parts, and the high-frequency parts:
\begin{align}
    \vA = \sqrt{\frac{{\beta}}{{2\sqrt{2\pi}}}}\L(\int_{\labs{\omega'}\le \Omega'}+\int_{\labs{\omega'}\ge \Omega'} \R)\hat{\vA}_{\vH'}(\omega')\rd \omega'.
\end{align}
We can then truncate the high-frequency part - hopefully with a small error - and keep the remaining low-frequency part that behaves nicely with the imaginary time evolution. Indeed, performing the operator fourier transform, applying the truncation scheme, we arrive at the following quasi-local version of Equation~\eqref{eq:identify_intro} (see \autoref{lem:truncating_bohr}):
\begin{align}
\braket{ \vO,[\vA,\vH-\vH']}_{\vrho}
    &= \frac{\text{const.}}{\beta} \int_{-\infty}^{\infty}\int_{-\infty}^{\infty}\tr\bigg[\vO^{\dagger}_{\vH}(t) \big( h_+(t')\vA_{\vH'}(t'+t)\vrho  - h_-(t')\vrho\vA_{\vH'}(t'+t)\big) \bigg] g_{\beta}(t)\rd t'\rd t\\
    &\quad + \text{(error terms)}, \label{eq:observableintro}
\end{align}
where $h_+, h_-$ are rapidly decaying functions.

\vspace{0.1in}

\noindent {\bf Solution to Issue 2:} Observe that the following variant of the RHS in Equation \eqref{eq:identify_intro}, the \textit{identifiabilty observable}, is in fact measurable in experiments: 
\begin{align}
Q(\vO,\vG,\vA,\vH'):=\frac{1}{\sqrt{2\pi}} \int_{-\infty}^{\infty}\int_{-\infty}^{\infty}\tr\bigg[\vO^{\dagger}_{\vG}(t) \big( h_+(t')\vA_{\vH'}(t'+t)\vrho  - h_-(t')\vrho\vA_{\vH'}(t'+t)\big) \bigg] g_{\beta}(t)\rd t'\rd t 
\label{eq:quasilocalobsintro}
\end{align}
where $\vG$ is any local Hamiltonian that the experimenter can choose. Properties of the identifiability observable are the key to our local learning algorithm and make transparent the local sufficient statistics property in our approach. Crucially, the expression still vanishes when $\vH=\vH'$, due to Equation \eqref{eq:detailedbalintro}, and stays negligible for \textit{any} $\vG$ when $\vH$ agrees locally with $\vH'$. 
 
Further, since $\vG, \vH'$ are local Hamiltonians, we can appeal to routine Lieb-Robinson bounds to significantly reduce relevant choices of $\vG,\vH'$ (See Appendix~\ref{sec:LRbounds}), by restricting to a radius $\ell$ around around $\vO$ (and $\vA$), while ensuring that $Q(\vO,\vG_\ell,\vA,\vH'_\ell)\approx Q(\vO,\vG,\vA,\vH')$. Here, $\vG_\ell$ (and $\vH'_\ell$) is the restriction of $\vG$ (and $\vH'$) to the ball of radius $\ell$ around $\vO$ (and $\vA$). 

\subsubsection{Protocol on general graphs}

\begin{figure}[t]
\includegraphics[width=0.9\textwidth]{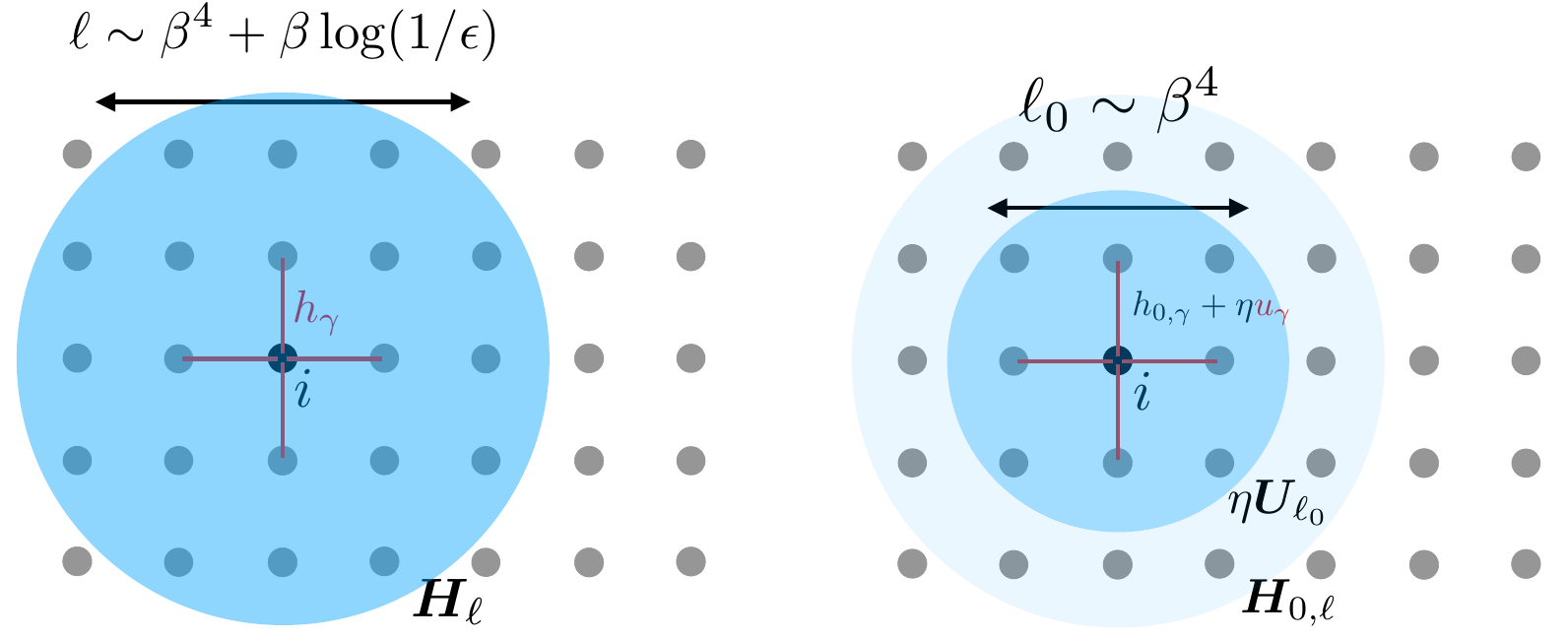}\caption{Our local learning protocols applied to 2-dimensional lattice with nearest neighbour interactions. (Left) The non-adaptive learning protocol. To learn a term $h_{\gamma}$ touching a site $i$, it suffices to search over local Hamiltonians $\vH_{\ell}$ for a distance (defined on the interaction graph,~\autoref{sec:Ham}) $\ell$ depending on the inverse temperature $\beta$ and error $\epsilon.$ Remarkably, we do not need access to regions far away from $i$, and the measurement only involves Heisenberg dynamics of local operators. Even though the algorithm is local for a fixed precision, the time complexity grows with search volume and is generally quasi-polynomial in $1/\epsilon$. 
(Right) An improved iterative learning protocol. Instead of achieving high-precision directly, we aim to double the precision each iteration: Given a pretty good guess $\vH_{0}$ such that $\vH = \vH_0 +\eta \vU$, we want to further refine $\vU$. This allows us to search only within a radius $\ell_0$ independent of the error.
}\label{fig:local_learning}
\end{figure}
We are now ready to sketch our algorithm, which performs a greedy local search to learn each term independently. For a site $i$, consider Pauli operators $\vA \in \{\vX_i,\vY_i,\vZ_i\}$, and $\vO = [\vA,\vP_{\gamma}]$ (which will be sufficient for bounding $\vO=[\vA, \vH-\vH']$). The strategy can be summarized in a sentence:

\begin{center}
\textit{Search for an $\vH'_\ell$ where $Q(\vO,\vG_\ell,\vA,\vH'_\ell)$ is small for all $\vG_\ell, \vA, \vO$. Record the terms overlapping with $i$.}    
\end{center}
 Such $\vH'_\ell$ always exist by setting $\vH'_\ell=\vH_\ell$ (as discussed in the previous subsection). Conversely, any such $\vH'_\ell$ is guaranteed to be locally unique, since setting $\vG_\ell=\vH_\ell$ implies the local commutator $[\vA, \vH'_\ell-\vH_\ell]$ needs to be small (see \autoref{sec:KMSfaithful} on the faithfulness of KMS norm) for all Paulis $\vA\in \{\vX_i,\vY_i,\vZ_i\}$ acting on $i$. This means that $\vH'_\ell$ and $\vH_\ell$ approximately agree for all the terms acting on site $i$. For a targeted error, we will choose $\ell = \CO(\beta^4+\beta \log(1/\epsilon))$ and put a suitably dense covering net for coefficients for $\vH'_{\ell}.$ 

The advertised scaling $\CO(\log n)$ for the sample complexity follows from measuring as many $Q$ in parallel as possible. The dependence on precision is far from $\frac{1}{\epsilon^2}$, but we can get near-optimal scaling on precision when restricted to $D$-dimensional Hamiltonians.

\subsubsection{Iterative protocol for lattices}

In the high-precision regime $\epsilon \rightarrow 0$, the radius of local neighborhood for which we truncate the time evolution $\vA_{\vH'}(t)$ needs to grow logarithmically with the error $\epsilon$, due to the Lieb-Robinson bounds. Thus, naively, the search space over all possible Hamiltonians $\vH'$ in the neighborhood of site $i$ still grows poorly with the precision $\sim \log(1/\epsilon)^{D}$, as discussed in the previous subsection.

To achieve the advertised near-optimal dependence on $\epsilon$, we propose an iterative protocol that gradually improves the precision in parallel sweeps. Let us assume we have already found a constantly-good (say $\eta=0.1$) candidate $\vH_0$ for the underlying Hamiltonian $\vH$ and we wish to further improve the precision by considering local Hamiltonians of the form $\vH_0+\eta \vU$. The observation is that, the identifiability observable~\eqref{eq:quasilocalobsintro} is actually most sensitive to terms near $i$ and, if the goal is to double the precision $\eta \rightarrow \eta/2,$ it suffices to search over a constant-sized (independent of error $\eta$) neighborhood near $i$. While the measurements still involve a $\log(1/\eta)^{D}$ sized neighborhood, the constant-sized neighborhood of the search space significantly saves on the sample complexity. We thus proceed by searching over $\vU'_\ell$ for regions of radius $\ell\sim \beta^4$ around $i$. The search completes with local terms of $\vU'_\ell$ and $\vU_\ell$ that act on site $i$ being $\frac{\eta}{2}$ close. Then we iterate the algorithm again up until $\eta=\epsilon$. This iterative algorithm is the only place where we heavily rely on the lattice geometry (where the number of terms only grows as $\sim \ell^{D}$ within distance $\ell$).

\subsection{Prior work}

\begin{table}[h]
    \centering
    \begin{tabular}{lcccc}  
        \toprule
         & Sample Complexity & Time Complexity &  Qubits entangled \\
         \hline
        \midrule
        \autoref{thm:latticethm} (Lattices) & $\CO\L(\log n\cdot\frac{e^{\poly(\beta)}}{\beta^2\epsilon^2}\poly(\log\frac{1}{\epsilon})\R)$ & $\CO\L(n\log n\cdot\frac{e^{\poly(\beta)}}{\beta^2\epsilon^2}\poly(\log\frac{1}{\epsilon})\R)$ & $\poly(\beta,\log\frac{1}{\epsilon})$ \\
        \autoref{thm:graphthm} (Graphs) & $\CO\L(\log n\cdot 2^{2^{\CO(\beta^4)}\poly(1/\beta\epsilon)}\R)$ & $\CO\L(n\log n\cdot 2^{2^{\CO(\beta^4)}\poly(1/\beta\epsilon)}\R)$ & $\poly(\beta,\log\frac{1}{\epsilon})$ \\
                \cite{BLMT24,narayanan2024improved} (Graphs) & $\poly\L(n,\frac{1}{\epsilon^{\CO(\beta^2)}}\R)$ & $\poly\L(n,\frac{1}{\epsilon^{\CO(\beta^2)}}\R)$ & $\CO\L(\beta^2\log\frac{1}{\epsilon}\R)$ \\
        \cite{HKT22} (High temp, Graphs) & $\CO\L(\log (n)\frac{1}{\beta^2\epsilon^2}\R)$ & $\CO\L(n\log (n)\frac{1}{\beta^2\epsilon^2}\R)$ & $\CO(\log\frac{1}{\epsilon})$ \\
        \cite{anshu2020sample} (Lattices) & $\poly(n)\frac{e^{\poly(\beta)}}{\poly(\beta)\epsilon^2}$ & $2^{\CO(n)}\cdot\frac{e^{\poly(\beta)}}{\poly(\beta)\epsilon^2}$  & $\CO(1)$\\
        \bottomrule
    \end{tabular}
    \caption{Comparison of different works based on sample complexity, time complexity, locality, and size of measurements, for success probability $0.99$. The time complexity combines classical computation costs and quantum gate complexity. Some results apply to Hamiltonians with bounded interaction degree (including expander graphs), while some to $D$-dimensional lattice Hamiltonians, and we use $\CO(\cdot)$ to suppress dependences on the geometric constants (degree $d$, dimension $D$, and locality $q$). The measurements in \cite{BLMT24} entangle far-away qubits (with respect to the graph distance), whereas the measurements in all the other works entangle nearby qubits of stated size. 
    }
    \label{tab:comparison}
\end{table}

While our approach aims at exposing new locality aspects in Hamiltonian learning from Gibbs states, we comment on some parallels with~\cite{BLMT24}. They also consider a KMS-like condition~\cite[Eq (1)]{BLMT24} reminiscent of our~\eqref{eq:detailedbalintro} and manage to show uniqueness through an involved sum-of-square argument. In our case, we were able to directly isolate a thermodynamically-inspired observable that identifies the local terms in a single analytic equation~\eqref{eq:identify_intro}. \cite[Eq. (2-4)]{BLMT24} also controls the imaginary-time evolved operator by low degree approximation with respect to $\vrho_\beta(\vH)$. We believe our operator Fourier transform may offer a transparent method to achieve a similar goal.

Perhaps a key difference is that they need measurements of faraway qubits\footnote{In \cite[Lemma 8.1]{BLMT24}, the authors want to bound an expression involving the commutator $[H,H']$. For this, they need item 3 in condition 5 on Page 30, for all $A_1,A_2$ of small size, including those $A_1,A_2$ that are far from each other in distance.
} that has no analog in our setting, which might be the obstruction for improving their sample complexity from polynomial to logarithmic in the system size. Our use of the operator Fourier transform $\vA_{\vH}(\omega)$ could potentially address this; it is plausible that plugging in Equation \eqref{eq:detailedbalanceloc} in the sum-of-squares techniques of \cite{BLMT24} would also lead to a local method, but we do not pursue this here.

As mentioned, other prior works either consider the high temperature regime \cite{HKT22}, or give up time efficiency of classical post-processing in favor of simple measurements \cite{anshu2020sample}. We reproduce a variance lower bound reminiscent of~\cite{anshu2020sample} to show that the KMS norm is locally faithful, but now using modern operator Fourier transform toolkits. Other works \cite{EyalAL19,LBAALA23} consider heuristic approaches, and it would be interesting to modify our approach towards very local measurement, but with semi-heuristic guarantees.

\subsection{Discussion and open problems}
We have shown how to learn Hamiltonian from its Gibbs state locally. This achieves a near-optimal sample and time complexity on all lattices for the dependence on the error $\epsilon$ and the system size $n$. Our work, on the one hand, completes the theoretical understanding of the learnability of this physically relevant class of Hamiltonians. On the other hand, it opens up a series of new questions. 

\begin{itemize}
    \item {\bf Reducing measurement cost:} Our algorithm still require measurements on $\CO(\beta^D)$ qubits, which can be large in practice. The work \cite{anshu2020sample} showed how to learn the $k$-local Hamiltonians with $k$-local measurements, albeit with very large time complexity. As further evidence, in the classical case, learning can be done with $\CO(1)$-locality. Suppose we know that $\vH$ is a 2-local classical Hamiltonian - for instance, the Ising model. To learn all the local terms in the neighborhood of a site $i$, we can leverage the Markov property of the classical Gibbs state - the distribution on site $i$ only depends on the spin configuration in the neighborhood of $i$. A local experiment can easily identify this conditional distribution via tomography. This can then be used to reconstruct the entire Hamiltonian.
    
    Is there a method that achieves near-optimal sample and time complexity, while performing entangling measurements on $\CO(1)$ sites? Some evidence in favour of entangling measurements on $\CO(\beta^D)$ qubits comes from the fact that the most recent Hamiltonian learning from time evolution $e^{i\vH t}$ also uses entangling measurements on $\Omega(t^D)$ qubits \cite{HKT22, HTFS23}. Since Gibbs states can inherently be viewed as imaginary time evolutions, $\CO(\beta^D)$ qubit entangling measurement seems fundamental.
    \item {\bf Near optimal protocol on more general graphs:} Is it possible to achieve near optimal sample and time complexity on more general family of interactions beyond lattices? In chemical and atomic systems, the interactions, strictly speaking, have a power-law decay, and a single electron could interact with an extensive number of particles with varying weights. It is also very natural to study Hamiltonians that live on general, expander graphs (such as the sparse SYK model \cite{HSHT23} or quantum Boltzmann machines). 
    \item \textbf{Structure learning:} For classical Hamiltonians, it is possible to learn the underlying graph structure with optimal sample and time complexities, under the promise that the graph has a low degree \cite{Bresler15,KlivansM17}. A crucial assumption in our Theorem \ref{thm:graphthm} is that the underlying interaction graph is known. Is it possible to learn the interaction graph itself, under the promise that it has a low degree?
\item \textbf{Connection to Markov properties.} 
The simplicity and locality of classical Gibbs state learning algorithms are intimately related to the Markov property of the distribution. A vertex, conditioned on its nearest neighbors, is independent from the remaining vertices, and one may resample the vertex of interest conditioned on the neighbors. Curiously, the very recent work~\cite{chen2025Markov} showed that quantum Gibbs states also satisfy a local Markov property, namely, local disturbance to the Gibbs state can be approximately recovered by running a (quasi)-local Gibbs sampler covering the region. Although we were not able to directly adapt the local Markov property for the present learning task, we did exploit the regularization argument, and the locality of quantum operator Fourier transform~\cite{chen2023efficient} appears to be a common theme. 
    \item {\bf Heisenberg scaling of error from thermofield double states:} Is it possible to reduce the error from $\CO(\frac{1}{\epsilon^2})$ to $\CO(\frac{1}{\epsilon})$ when the purification to the Gibbs state (called the thermofield double state) is available? This is possible when $\vH$ is classical, as one has access to the purifications of conditional probability distributions.   
\end{itemize}

\subsection*{Roadmap}
We begin with the preliminaries, including the KMS inner product, the Hamiltonian family together with the interaction graph, the operator Fourier transform, and a regularization trick at low temperatures. Next, we expand on the key analytic argument circling around the identifiability equation. We conclude with the main learning protocols that apply the identifiability equation. In the appendix, we allocate arguments less central to the main conceptual message, including standard Lieb-Robinson bounds, truncation bounds, and standard quantum algorithm subroutines. 
\section{Preliminaries}
\subsection*{Notations}\label{sec:recap_notation}
Throughout the paper, we write
\begin{align}
a \lesssim b \quad \text{iff}\quad a \le c b \quad \text{for an absolute constant} \quad c>0.
\end{align} 
 We denote asymptotic upperbound $\CO(\cdot)$ when fixing the geometric parameters of the Hamiltonian (degree $d$, dimensional $D$, and locality $q$), and $\poly(\cdot)$ denotes a polynomial depending only on $q,d,D$. We use $\bigOt{\cdot}$ to further absorb (poly)logarithmic factors.  
We write scalars, functions, and vectors in normal font, and matrices in bold font $\vO$. 

\begin{align}
\vI&: &\text{the identity operator}\\
\beta&: &\text{ inverse temperature}\\
\vrho_{\beta}&:= \frac{\e^{-\beta \vH }}{\tr[ \e^{-\beta \vH }]} (\equiv \vrho) \quad &\text{the Gibbs state with inverse temperature $\beta$}\\
n &= \labs{\Lambda} &\text{ system size (number of qubits) of the Hamiltonian $\vH$}
\end{align}
Fourier transform notations:
\begin{align}
\vH &= \sum_i E_i \ketbra{\psi_i}{\psi_i}&\text{the Hamiltonian of interest and its eigendecomposition}\\
\text{Spec}(\vH) &:= \{ E_i \} & \text{the spectrum of the Hamiltonian}\\
\nu \in B = B(\vH)&:= \text{Spec}(\vH) - \text{Spec}(\vH) &\text{the set of Bohr frequencies}\\
\vP_{E}&:= \sum_{i:E_i = E} \ketbra{\psi_i}{\psi_i}&\text{eigenspace projector for energy $E$}\\
\vA_\nu&:=\sum_{E_2 - E_1 = \nu } \vP_{E_2} \vA \vP_{E_1} &\text{operator $\vA$ at exact Bohr frequency $\nu$}\\
{\vA}_{\vH}(t) &:=\e^{i\vH t}\vA \e^{-i\vH t}& \text{time-evolved operator $\vA$ with $\vH$}\\
\hat{\vA}_{\vH}(\omega) &:= \frac{1}{\sqrt{2\pi}}\int_{-\infty}^{\infty} \e^{-\ri \omega t}f(t) \vA(t)\mathrm{d}t& \text{operator Fourier Transform for $\vA$ weighted by $f$}\\
\hat{f}(\omega)&=\lim_{K\rightarrow  \infty}\frac{1}{\sqrt{2\pi}}\int_{-K}^{K}\e^{-\ri\omega t} f(t)\mathrm{d}t & \text{the Fourier transform of a function $f$}		
\end{align}
Norms: 
\begin{align}
	\norm{\vO}&:= \sup_{\ket{\psi},\ket{\phi}} \frac{\bra{\phi} \vO \ket{\psi}}{\norm{\ket{\psi}}\cdot \norm{\ket{\phi}}}= \norm{\vO}_{\infty} \quad &\text{the operator norm of a matrix $\vO$}\\
  	\norm{\vO}_p&:= (\tr \labs{\vO}^p)^{1/p} \quad&\text{the Schatten p-norm of a matrix $\vO$}\\
    \langle \vX,\vY\rangle_{\vrho}&:=\tr[\vX^\dagger \vrho^{\frac{1}{2}}\vY\vrho^{\frac{1}{2}}]\,.&\text{the Kubo-Matrin-Schwinger inner product} 
\end{align}

\subsection{Gibbs state and KMS inner product}
We recall that, given a full-rank state $\vrho$, the KMS inner product of two operators $\vX,\vY$ is
\begin{align}
\langle \vX,\vY\rangle_{\vrho}:=\tr[\vX^\dagger \vrho^{\frac{1}{2}}\vY\vrho^{\frac{1}{2}}]\,.
\end{align}
In this paper, we denote by 
\begin{align}
    \norm{\vX}_{\vrho} := \sqrt{\langle \vX,\vX\rangle_{\vrho}}
\end{align}
the $\vrho$-weighted norm induced by the KMS inner product. In particular, we will only consider the Gibbs state $\vrho= \e^{-\beta \vH}/\tr[\e^{-\beta \vH}]$ associated the ground truth Hamiltonian $\vH.$ 
\begin{rmk}
 Our current argument implicitly requires the KMS inner product and does not obviously work for other choices, such as the GNS inner product.
\end{rmk}
The conversion to operator norm always holds, but sometimes may be suboptimal.
\begin{lem}[Operator norm controls weighted norms and inner-product]\label{lem:operatornorm}
    Unconditionally, we have that $\norm{\vX}_{\vrho}\le \norm{\vX}$ and $\braket{\vX,\vY}_{\vrho} \le \norm{\vX}\norm{\vY}.$
\end{lem}

\subsection{Hamiltonians on bounded degree interaction graph and on lattices}\label{sec:Ham}

On a set $\Lambda$ of $n = \labs{\Lambda}$ qubits, we consider Hamiltonians $\vH$ where each term $\vh_{\gamma}$ acts on at most $q$-qubits 
\begin{align}
    \vH = \sum_{\gamma\in \Gamma} \vh_{\gamma}\quad \text{where}\quad \norm{\vh_{\gamma}} \le 1.
\end{align}
From this decomposition, we define the interaction graph\footnote{Every Hamiltonian also defines a hypergraph when the hyperedges are the $\gamma$. Here, the interaction graph is defined between the Hamiltonian terms $\vh_{\gamma}$.} with vertices corresponding to the set $\Gamma$, and we draw an edge between $\gamma_1$ and $\gamma_2$ if and only if the terms have overlapping supports (self-loop allowed):
\begin{align}
\gamma_1\sim \gamma_2\quad \iff\quad    \text{Supp}(\vh_{\gamma_1}) \cap \text{Supp}(\vh_{\gamma_2}) \ne \emptyset.
\end{align}
Similarly, we may consider any subset of vertices $A\subset \Lambda$ and write 
\begin{align}
    A \sim \gamma \iff A\cap \text{Supp}(\vh_{\gamma})\ne \emptyset.
\end{align}
The maximal \textit{degree} of the interaction graph is denoted by $d$, and we are particularly working in the regime where $d$ is a constant independent of the system size $n$\footnote{This will ensure the possibility of conjugating by a constant temperature Gibbs state (see \autoref{lem:convergence_imaginary})}. 
For any two subsets of vertices $A,B\subset \Lambda$, we denote by $\text{dist}(A,B)$ the minimal length of a path connecting $A$ to $B$ via interactions in $H$:
\begin{align*}
\text{dist}(A,B)=\min\bigg\{\ell\in\mathbb{N}:\exists \gamma_1,\dots\gamma_{\ell}\in\Gamma\quad \text{such that}\quad A\sim \gamma_1\sim\gamma_2\sim\dots \sim \gamma_{\ell}\sim B\bigg\}\,.
\end{align*}
Often, we will also consider the subset $A$ or $B$ to be the Hamiltonian term $\gamma$, and we will simply abuse the notation to write  $\text{dist}(\gamma,\gamma')$ and $\text{dist}(A,\gamma')$. We will also often refer to the support of an operator $\vA$ and write by $\text{dist}(\vA,\gamma') \equiv \text{dist}(\text{Supp}(\vA),\gamma').$ For a subset $A\subset \Lambda$, we often consider the local Hamiltonian patch $\vH_{\ell}$ containing all terms $\vh_{\gamma}$ with distance $\text{dist}(\gamma,A)< \ell-1$ 
\begin{align}
    \vH_{\ell} = \sum_{\gamma:\text{dist}(\gamma,A) < \ell-1} \vh_{\gamma}.
\end{align}
Let us also define a surface and volume associated with a ball around set $A$
\begin{align}
    S(\ell,A):= \labs{\{\gamma:\text{dist}(\gamma,A) = \ell\}},\\
    V(\ell,A):=\labs{\{\gamma:\text{dist}(\gamma,A) \le \ell\}}.
\end{align}
We always have $S(\ell) \le \labs{A} d^{\ell+1}, V(\ell) \le \labs{A} d^{\ell+2}/(d-1) \le \labs{A} d^{\ell+2}.$ 

Some of our results consider the special case of $D$-dimensional Hamiltonians. For our proofs, we will simply define a family of $D$-dimensional lattice Hamiltonians by having a uniform bound on the volume and area around a set $A$ 
\begin{align}
    S(\ell)&\le \CO( |A|\ell^{D-1}),\\
    V(\ell)&\le \CO( |A| \ell^D) 
\end{align}
and that the degree $d$ and locality $q$ are constants and $\CO(\cdot)$ suppress dependence on $D,d,q$. These cover the case of nearest-neighbour lattice and a more general ``finite-range'' interaction that allows for arbitrary localized interaction. For concreteness in describing the learning protocol, we will ultimately think of each term as a distinct nontrivial\footnote{Distinctness guarantees the coefficients to be unique for the same Hamiltonian. Removing the identity operator ensures that the Gibbs state uniquely determines the Hamiltonian when $ \beta > 0$.} Pauli string $\vP_{\gamma}\in \{\vI,\vX,\vY,\vZ\}^{\otimes n}$ up to $q$-body, weighted by scalar $h_{\gamma}$
\begin{align}
    \vh_{\gamma} = h_{\gamma} \vP_{\gamma}\quad \text{where}\quad h_{\gamma}\in [-1,1]\quad  \text{and}\quad \norm{\vP_{\gamma}} = 1.
\end{align}

However, some of the basic subroutines can be stated only with the geometry of the interaction graph, without committing to a particular representation. 

\subsection{Operator Fourier tranforms}
\label{sec:OFT}
We recall the operator Fourier transform~\cite{chen2023quantum,chen2023efficient} of an operator $\vA$ associated to the Hamiltonian $\vH$ with spectral decomposition $\vH=\sum_{i}E_i\vP_{E_i}$,

\begin{align}\label{eq:OFT}
{\hat{\vA}_{\vH}}(\omega)=  \frac{1}{\sqrt{2\pi}}\int_{-\infty}^{\infty} \e^{\ri \vH t} \vA \e^{-\ri \vH t} \e^{-\ri \omega t} f(t)\rd t
    \end{align}
with a Gaussian weight with an energy width $\sigma>0$
    \begin{align}
        \hat{f}(\omega)=\frac{1}{\sqrt{\sigma\sqrt{2\pi}}} \exp\L(- \frac{\omega^2}{4\sigma^2}\R),\quad \text{and}\quad f(t) = e^{-\sigma^2t^2}\sqrt{\sigma\sqrt{2/\pi}}\label{eq:fwft}
    \end{align}
    such that $\int_{-\infty}^{\infty}\labs{f(t)}^2\rd t =\int_{-\infty}^{\infty}\labs{\hat{f}(\omega)}^2\rd \omega=1.$ Recall the Fourier transform pairs
    \begin{align}
        \hat{f}(\omega)=\frac{1}{\sqrt{2\pi}}\int_{-\infty}^{\infty}\e^{-\ri\omega t} f(t)\mathrm{d}t\quad \text{and}\quad f(t)=\frac{1}{\sqrt{2\pi}}\int_{-\infty}^{\infty}\e^{\ri\omega t} \hat{f}(\omega)\mathrm{d}\omega. 
    \end{align}
However, in this paper, we will reserve $f(t)$ exclusively for the Gaussian weight. A key object is the decomposition of an operator $\vA$ by the Bohr frequencies $\nu\in B(\vH)$ of a Hamiltonian $\vH$
\begin{align}
        \vA = \sum_{\nu\in B(\vH)} \vA_{\nu},\quad \text{where}\quad \vA_\nu&:=\sum_{E_2 - E_1 = \nu } \vP_{E_2} \vA \vP_{E_1} \quad \text{satisfies that} \quad (\vA_{\nu})^{\dagger} = (\vA^{\dagger})_{-\nu},\label{eqn:Aoperator}
\end{align}
and $\vP_{E}$ are eigenspace projectors for energy $E\in \text{Spec}(\vH)$ and $B(\vH)=\text{Spec}(\vH)-\text{Spec}(\vH)$ is the set of energy differences. 

\begin{lem}[Decomposing into Bohr frequencies~{\cite[Appendix A]{chen2023quantum}}]\label{lem:Bohr_decomp} For any Hamiltonian $\vH$, the Heisenberg dynamics for a (not necessarily Hermitian) operator $\vA$ can be decomposed as
\begin{align}
    \vA_{\vH}(t):=\e^{\ri \vH t} \vA \e^{-\ri \vH t} &= \sum_{\nu\in B(\vH)} \e^{\ri \nu t} \vA_\nu.
\end{align}
Furthermore, the operator Fourier transform satisfies
\begin{align}
   {\hat{\vA}_{\vH}}(\omega) =\sum_{\nu \in B(\vH)} \vA_{\nu}\hat{f}(\omega - \nu).
\end{align}
\end{lem}

\subsection{Regularizing the operator Fourier transform at low-temperatures}\label{sec:regularize}

At a low constant temperature, the imaginary time dynamics can diverge exponentially with the system size $\norm{e^{\beta\vH} \vA e^{-\beta \vH}} \ge e^{cn}$ in more than one spatial dimension \cite{Bouch15}. 
It will be tremendously helpful to decompose the operator over operator Fourier transforms at different Bohr frequencies. This section follows results from~\cite{chen2025Markov}, and we include some of the relevant proofs. 
\begin{lem}[Decomposing an operator by the energy change{~\cite[Lemma IX.1]{chen2025Markov}}]\label{lem:sumoverenergies}
For any (not necessarily Hermitian) operator $\vA$ and Hermitian $\vH$, we have that
\begin{align}
\vA =  \frac{1}{\sqrt{\sigma2\sqrt{2\pi}}} \int_{-\infty}^{\infty} \hat{\vA}_{\vH}(\omega)\rd \omega.
\end{align}
\end{lem}
\begin{proof}
    \begin{align}\vspace{-1cm}
        \int_{-\infty}^{\infty} \hat{\vA}_{\vH}(\omega)\rd \omega = \int_{-\infty}^{\infty} \sum_{\nu} \vA_{\nu} \hat{f}(\omega-\nu) \rd \omega= \sum_{\nu} \vA_{\nu} \int_{-\infty}^{\infty}\hat{f}(\omega-\nu) \rd (\omega-\nu)
        = \vA \sqrt{2\pi} f(0) = \vA \sqrt{\sigma2\sqrt{2\pi}}.
    \end{align}
    Rearrange to conclude the proof.
\end{proof}

The Gaussian damping has a regularization effect due to its super-exponential decay. 
\begin{lem}[Norm bounds on imaginary time conjugation{~\cite[Lemma IX.2]{chen2025Markov}}]\label{lem:bounds_imaginary_conjugation}
For any $\beta ,\omega\in \BR$ and operator $\vA$ with norm $\norm{\vA}\le 1$, the operator Fourier transform $\hat{\vA}_{\vH}(\omega)$ with uncertainty $\sigma$~\eqref{eq:OFT},~\eqref{eq:fwft} satisfies 
\begin{align}
        \e^{\beta \vH} \hat{\vA}_{\vH}(\omega) \e^{-\beta \vH}
& = e^{\beta\omega} \cdot \hat{\vA}_{\vH}(\omega+2\sigma^2\beta) \e^{\sigma^2\beta^2}.
\end{align}
Thus, 
\begin{align}
    \norm{\e^{\beta \vH} \hat{\vA}_{\vH}(\omega) \e^{-\beta \vH}} \le \frac{\e^{\sigma^2\beta^2}}{\sqrt{{\sigma}\sqrt{2\pi}}} \e^{\beta\omega}.
\end{align}
\end{lem}
In comparison, directly conjugating the unfiltered operator could yield a norm $\norm{\e^{\beta \vH} \vA \e^{-\beta \vH}}$ growing with the system size $n$; the Gaussian filtering centered at Bohr frequency $\omega$ removes the dependence on the system size $n$, and only depends the Bohr frequency $\omega.$ While it still grows exponentially, the bounds are now entirely (quasi)-local.
\begin{proof}
Recall
\begin{align}
        \e^{\beta \vH} \hat{\vA}_{\vH}(\omega) \e^{-\beta \vH} &= \sum_{\nu} \vA_{\nu}\frac{1}{\sqrt{\sigma\sqrt{2\pi}}} \exp\L(- \frac{(\omega-\nu)^2}{4\sigma^2}\R)\e^{\beta \nu} \tag*{(\autoref{lem:Bohr_decomp})} \\
&=  \sum_{\nu} \vA_{\nu}\frac{1}{\sqrt{\sigma\sqrt{2\pi}}} \exp\L(- \frac{(\omega+2\sigma^2\beta-\nu)^2}{4\sigma^2}+\beta\omega+\sigma^2\beta^2\R) \\
& = \hat{\vA}(\omega+2\sigma^2\beta) \cdot \e^{\beta\omega+\sigma^2\beta^2}.
\end{align}
Apply triangle inequality to the time integral $\norm{\hat{\vA}(\omega)} \le\frac{1}{\sqrt{2\pi}} \int_{-\infty}^{\infty}\labs{f(t)}\rd t =\hat{f}(0) = \frac{1}{\sqrt{\sigma\sqrt{2\pi}}}$ to conclude the proof.
\end{proof}
At high enough temperatures, there is a stronger bound (within the convergence radius of the Taylor expansion) that exploits the bounded interaction degree of the Hamiltonian.

\begin{lem}[Convergence for imaginary time{~\cite[Lemma IX.3]{chen2025Markov}}]\label{lem:convergence_imaginary}
    For Hamiltonians defined in~\autoref{sec:Ham} with interaction degree at most $d$, and a single-site operator $\vA$, and $\labs{\beta}< 1/2d,$
    \begin{align}
        \norm{\e^{\beta \vH}\vA\e^{-\beta \vH}}\le \frac{1}{1-2d\labs{\beta}}.
    \end{align}
\end{lem} 
Using the above, we bootstrap for an even better bound on the norm of the Operator Fourier Transform.
\begin{cor}[Norm decay for large energy difference{~\cite[Corollary IX.2]{chen2025Markov}}]\label{cor:norm_decay}
For any $\beta ,\omega \in \BR$ and operator $\vA$, the operator Fourier transform with uncertainty $\sigma>0$ satisfies 
\begin{align}
    \norm{\hat{\vA}_{\vH}(\omega)} \le \frac{\e^{-\beta \omega + \sigma^2 \beta^2}}{\sqrt{{\sigma}\sqrt{2\pi}}} \norm{\e^{\beta \vH}\vA\e^{-\beta \vH}}.
\end{align}
\end{cor}
\begin{proof} ``Borrow'' cancelling factors of $e^{\beta \vH}$ on the left and right
    \begin{align}
\hat{\vA}_{\vH}(\omega)  &= \e^{-\beta \vH}\cdot (\e^{\beta \vH}\hat{\vA}_{\vH}(\omega)\e^{-\beta \vH}) \cdot\e^{\beta \vH}\\
&=\e^{-\beta \vH}\cdot( \widehat{[\e^{\beta \vH}\vA\e^{-\beta \vH}]}_{\vH}(\omega) )\cdot \e^{\beta \vH}\tag*{(Operator FT commutes with imaginary time conjugation)}
\end{align}    
and apply~\autoref{lem:bounds_imaginary_conjugation} for $\vA'= \e^{\beta \vH}\vA\e^{-\beta \vH}$ to conclude the proof.
\end{proof}
This will allow us to truncate an operator in the Bohr-frequency space with an exponentially small error.

\section{The identifiability equation}
In this section, we are interested in the local difference between the ground truth $\vH = \sum_{\gamma\in \Gamma} \vh_{\gamma}$ and a guess 
\begin{align}
    \quad\vH' = \sum_{\gamma'\in \Gamma} \vh_{\gamma'}\quad \text{where}\quad \norm{\vh_{\gamma'}} \le 1,
\end{align}
by a positive quantity 
\begin{align}
\norm{[\vA,\vH-\vH']}^2_{\vrho} \quad \text{for}\quad \vA \in \{\vX_i,\vY_i,\vZ_i\}.
\end{align}

The challenge is to control this quantity without \textit{a priori} knowing the ground truth $\vH$. The first step is to relax this expression by optimizing over a set of operators
\begin{align}
    \lnorm{[\vA,\vH-\vH']}^2_{\vrho} = \Braket{ [\vA,\vH-\vH'],[\vA,\vH-\vH'] }_{\vrho}  \le 2d \sup_{\vO = [\vA,\vP_{\gamma}] }\big\lvert \braket{ \vO,[\vA,\vH-\vH']}_{\vrho}\big\rvert. \label{eq:relaxing_square}
\end{align}
That is, we consider $\vO$ to be all possible terms when taking the commutators of $\vH$ with $\vA$.

The main result of this section is the following identifiability equation, whose precise functional form will be the key to local learning.
\begin{lem}[The identifiability equation]\label{lem:identify}
For any pair of Hamiltonians $\vH,\vH'$ and Gibbs states $\vrho \propto e^{-\beta \vH}, \vrho' \propto e^{-\beta \vH'}$, and operators $\vO$, $\vA$, we have that
    \begin{align}
    \frac{\beta}{2} \braket{ \vO,[\vA,\vH-\vH']}_{\vrho} = \frac{1}{\sqrt{2\pi}}\int_{-\infty}^{\infty} \tr\L[\vO^{\dagger}_{\beta\vH/2}(t)\L(  \sqrt{\vrho'}\vA_{\beta \vH'/2}(t)\sqrt{\vrho^{'-1}}\vrho - \vrho\sqrt{\vrho^{'-1}}\vA_{\beta \vH'/2}(t)\sqrt{\vrho^{'}} \R)\R]g(t) \rd t.
\end{align}
\end{lem}
The rest of the section begins with proof and supplies additional regularization tricks that will turn the RHS into physically measurable quantities with decent continuity properties. We will also derive the conversion between the commutator square and the local coefficients.  

\subsection{Double Bohr frequency decomposition}\label{sec:doubleBohr}

To make sense of~\autoref{lem:identify}, we must consider two Hamiltonians $\vH_1$, $\vH_2$ and decompose them into their Bohr frequencies $\nu_1,\nu_2$ iteratively, resulting in the following ``double'' decomposition, which may seem intimidating at first glance
\begin{align}
(\vA_{\nu_1})_{\nu_2} := \sum_{E'_2 - E_2 = \nu_2 } \sum_{E'_1 - E_1 = \nu_1 } \vP_{E'_2} \vP_{E'_1} \vA \vP_{E_1}\vP_{E_2}.
\end{align}
In general, the order of decomposition matters $(\vA_{\nu_1})_{\nu_2} \ne (\vA_{\nu_2})_{\nu_1}$ as $\vH_1$ and $\vH_2$ may not commute. Nevertheless, for the expressions we care about, their double Bohr frequency decomposition still takes a manageable form.

\begin{lem}[Double Bohr frequency decomposition]\label{lem:double_Bohr}
For any operator $\vA,$ and Hermitian operators $\vH_1,\vH_2$,
    \begin{align}
 \e^{\vH_2}\e^{-\vH_1}\vA\e^{\vH_1}\e^{-\vH_2}-\e^{-\vH_2}\e^{\vH_1}\vA\e^{-\vH_1}\e^{\vH_2} &= \sum_{\nu_1\in B_1,\nu_2\in B_2} (\vA_{\nu_{1}})_{\nu_2} 2\sinh(\nu_2-\nu_1),\label{eq:Osinh}\\
         [\vA,\vH_2] - [\vA,\vH_1] &= -\sum_{\nu_1\in B_1,\nu_2\in B_2} (\vA_{\nu_1})_{\nu_2} (\nu_2-\nu_1)\label{eq:OHH},
    \end{align}
    where $B_1,B_2$ are respectively the set of Bohr frequencies of $\vH_1,\vH_2$.
\end{lem}
\begin{proof}
Rewrite in the Bohr frequency basis
\begin{align}
    \e^{\vH_2}\e^{-\vH_1}\vA\e^{\vH_1}\e^{-\vH_2} =      \e^{\vH_2}(\sum_{\nu_1\in B_1}\vA_{\nu_{1}}\e^{-\nu_1})\e^{-\vH_2} =  \sum_{\nu_1\in B_1,\nu_2\in B_2} (\vA_{\nu_{1}}\e^{-\nu_1})_{\nu_2} \e^{\nu_2}= 
    \sum_{\nu_1\in B_1,\nu_2\in B_2} (\vA_{\nu_{1}})_{\nu_2} \e^{\nu_2-\nu_1}.
\end{align}
Similarly,
\begin{align}
    \e^{-\vH_2}\e^{\vH_1}\vA\e^{-\vH_1}\e^{\vH_2} = \sum_{\nu_1\in B_1,\nu_2\in B_2} (\vA_{\nu_{1}})_{\nu_2} \e^{\nu_1-\nu_2},
\end{align}
and take the difference to obtain the first line. Next, 
\begin{align}
    [\vA,\vH_1] &= -\sum_{\nu_1\in B_1} \vA_{\nu_1} \nu_1 \\
    &= -\sum_{\nu_2\in B_2}\sum_{\nu_1\in B_1} (\vA_{\nu_1})_{\nu_2} \nu_1. \tag*{(\autoref{lem:Bohr_decomp})}
\end{align}
And,
\begin{align}
    [\vA,\vH_2] &= -\sum_{\nu_2\in B_2} \vA_{\nu_2} \nu_2 \\
    &=-\sum_{\nu_2\in B_2} (\sum_{\nu_1 \in B_1}\vA_{\nu_1})_{\nu_2} \nu_2 \tag*{(\autoref{lem:Bohr_decomp})}.
\end{align}
In both cases, we decompose by $\nu_1$ in the inner layer and $\nu_2$ outside so that both expressions are linear combinations of $(\vA_{\nu_1})_{\nu_2}$. Take the difference to obtain the second line.
\end{proof}
Remarkably, the coefficients of both expressions in~\autoref{lem:double_Bohr} depends only on the difference $\nu_1,\nu_2.$

\begin{lem}[Rewriting commutator difference in the time domain]
\label{lem:comm_diff}
For any operator $\vA,$ and Hermitian operators $\vH_1,\vH_2$,
    \begin{align}
        [\vA,\vH_2] - [\vA,\vH_1] = \frac{1}{\sqrt{2\pi}}\int_{-\infty}^{\infty} \L[\e^{\vH_2}\e^{-\vH_1}\vA_{\vH_1}(t)\e^{\vH_1}\e^{-\vH_2}-\e^{-\vH_2}\e^{\vH_1}\vA_{\vH_1}(t)\e^{-\vH_1}\e^{\vH_2}\R]_{\vH_2}(-t) g(t) \rd t
    \end{align}
    where 
    \begin{align}
g(t) = -\frac{\pi^{3/2}}{2\sqrt{2}(1+\cosh(\pi t))}\quad \text{and}\quad        \hat{g}(\nu) := \frac{-\nu}{2\sinh(\nu)}.
    \end{align}
\end{lem}
\begin{proof}
We begin with decomposing into the Bohr frequencies by~\autoref{lem:double_Bohr} 
\begin{align}
    [\vA,\vH_2]-[\vA,\vH_1] &=  -\sum_{\nu_1\in B_1,\nu_2\in B_2} (\vA_{\nu_1})_{\nu_2} (\nu_2-\nu_1) \\
& = \sum_{\nu_1\in B_1,\nu_2\in B_2} (\vA_{\nu_{1}})_{\nu_2} 2\sinh(\nu_2-\nu_1) \hat{g}(\nu_2-\nu_1)\\
& = \sum_{\nu_1\in B_1,\nu_2\in B_2} (\vA_{\nu_{1}})_{\nu_2} 2\sinh(\nu_2-\nu_1) \frac{1}{\sqrt{2\pi}}\int_{-\infty}^{\infty} g(t) \e^{- \ri (\nu_2-\nu_1)t} \rd t\\
& =  \frac{1}{\sqrt{2\pi}}\int_{-\infty}^{\infty} \sum_{\nu_1\in B_1,\nu_2\in B_2} (\vA_{\nu_{1}}\e^{\ri\nu_1t} )_{\nu_2} 2\sinh(\nu_2-\nu_1) \e^{- \ri \nu_2t} g(t) \rd t.
\end{align}
Express the operator Fourier transforms in the time domain by~\autoref{lem:Bohr_decomp} to conclude the proof.
\end{proof}
\begin{rmk}
    When the Hamiltonians $\vH_1,\vH_2$ are very different, we expect the Gibbs expression (e.g., $\e^{\vH_2}\e^{-\vH_1}\vA_{\vH_1}(t)\e^{\vH_1}\e^{-\vH_2}$) to be very large in operator norm. However, after careful filtering, only the LHS remains. Indeed, the filter in the frequency domain $\hat{g}(\nu)$ decays exponentially with $\labs{\nu}$.
\end{rmk}

\begin{rmk}
There is information in the~\eqref{eq:Osinh} not present in~\eqref{eq:OHH}. For example, when the two global Hamiltonians $\vH_1,\vH_2$ are the same near a local operator $\vA$ but different elsewhere, the local commutator vanishes $[\vA,\vH_2] - [\vA,\vH_1] =0$, while the global $\e^{\vH_2}\e^{-\vH_1}\vA\e^{\vH_1}\e^{-\vH_2}-\e^{-\vH_2}\e^{\vH_1}\vA\e^{-\vH_1}\e^{\vH_2}$ may not. On the other hand, the coefficients of $(\vA_{\nu_1})_{\nu_2}$ are related in a bijection $x \rightarrow -2\sinh(x)$ for $x = \nu_2-\nu_1$. There is no contradiction: to access $(\vA_{\nu_1})_{\nu_2}$, we need to apply Hamiltonian dynamics from both inside and outside $([\vA(t_1)_{\vH_1},\vH_2] - [\vA_{\vH_1}(t_1),\vH_1])_{\vH_2}(t_2),$ which contains more information than in the commutator~\eqref{eq:OHH}.
\end{rmk}

\subsection{Relaxing a local commutator}
Now, invoke~\autoref{lem:comm_diff} to rewrite the commutator by imaginary-time evolutions to prove the key identifiability equation.

\begin{proof}[Proof of~\autoref{lem:identify}]
     Apply~\autoref{lem:comm_diff} with rescaled Hamiltonians $\vH_2\leftarrow \beta\vH$, $\vH_1\leftarrow \beta\vH'$ to obtain
\begin{align}
    &\frac{\beta}{2} [\vA,\vH-\vH'] \\
    &= \frac{1}{\sqrt{2\pi}}\int_{-\infty}^{\infty} \L(\e^{\beta\vH/2}\e^{-\beta\vH'/2}\vA_{\beta\vH'/2}(t)\e^{\beta\vH'/2}\e^{-\beta\vH/2}-\e^{-\beta\vH/2}\e^{\beta\vH'/2}\vA_{\beta\vH'/2}(t)\e^{-\beta\vH'/2}\e^{\beta\vH/2}\R)_{\beta\vH/2}(-t) g(t) \rd t \\
&= \frac{1}{\sqrt{2\pi}}\int_{-\infty}^{\infty} \L(\sqrt{\vrho^{-1}}\sqrt{\vrho^{'}}\vA_{\beta\vH'/2}(t)\sqrt{\vrho^{'-1}}\sqrt{\vrho}-\sqrt{\vrho}\sqrt{\vrho^{'-1}}\vA_{\beta\vH'/2}(t)\sqrt{\vrho^{'}}\sqrt{\vrho^{-1}}\R)_{\beta\vH/2}(-t) g(t) \rd t.\label{eq:AHH'}
\end{align}
At each time $t$, the $\vrho$-weighted expectation $ \braket{\vO,\cdot}_{\vrho} =\tr[\vO^{\dagger}\sqrt{\vrho}\cdot\sqrt{\vrho}]$ of the integrand yields
\begin{align}
&\tr\L[ \sqrt{\vrho}\vO^{\dagger}\sqrt{\vrho}\L(\sqrt{\vrho^{-1}}\sqrt{\vrho^{'}}\vA_{\beta\vH'/2}(t)\sqrt{\vrho^{'-1}}\sqrt{\vrho}-\sqrt{\vrho}\sqrt{\vrho^{'-1}}\vA_{\beta\vH'/2}(t)\sqrt{\vrho^{'}}\sqrt{\vrho^{-1}}\R)_{\beta\vH/2}(-t)\R]\\
&=\tr\L[\vO^{\dagger}\L(  \sqrt{\vrho'}\vA_{\beta \vH'/2}(t)\sqrt{\vrho^{'-1}}\vrho - \vrho\sqrt{\vrho^{'-1}}\vA_{\beta \vH'/2}(t)\sqrt{\vrho^{'}} \R)_{\beta\vH/2}(-t)\R]\\
&=\tr\L[\vO^{\dagger}_{\beta\vH/2}(t)\L(  \sqrt{\vrho'}\vA_{\beta \vH'/2}(t)\sqrt{\vrho^{'-1}}\vrho - \vrho\sqrt{\vrho^{'-1}}\vA_{\beta \vH'/2}(t)\sqrt{\vrho^{'}} \R)\R],
\end{align}
using that the outer time-dynamics $(\cdot)_{\beta \vH/2}$ commutes with the Gibbs state $\vrho\propto e^{-\beta \vH}$ and that $\tr[\vA(-t) \vB ]= \tr[\vA \vB(t)]$. Restore the integral to conclude the proof. 
\end{proof}
One may wonder why the RHS is any better than the LHS, as both depend on $\vH$ and $\vH'$. However, observe that
    \begin{align}
        \vH'=\vH,\quad \text{implies that}\quad \L(  \sqrt{\vrho'}\vA_{\beta \vH'/2}(t)\sqrt{\vrho^{'-1}}\vrho - \vrho\sqrt{\vrho^{'-1}}\vA_{\beta \vH'/2}(t)\sqrt{\vrho^{'}} \R) =0\quad \text{for each}\quad t.
    \end{align}
    Therefore, there must exist a guess $\vH'$ such that the RHS~\eqref{eq:relaxing_square} vanishes for all $\vO$. Moreover, we can \textit{verify} that the RHS is zero by enumerating all possible $\vO_{\beta \tilde{\vH}}(t),$ without knowing $\vH$ apriori. To make this observation quantitative, we will need to further massage the RHS, by suitably regularizing the Gibbs conjugation as the following section, and understanding the effect of errors when the guess $\vH'$ is not exactly $\vH$ (\autoref{lem:unique}). 
    \begin{rmk}
        \autoref{lem:identify} seems implicitly tied with KMS inner product; we do not know how to replicate the argument for GNS inner product. Indeed, in~\eqref{eq:AHH'}, changing the ordering or exponent of $\vrho,\vrho'$, might either lose locality or fail to recover non-fractional powers of $\vrho$ on the RHS.
    \end{rmk}
\subsection{Regularizing high-frequency parts}
However, in the identifiability equation (\autoref{lem:identify}), the conjugation of Gibbs state on the RHS, while the LHS is always bounded. To regularize this divergence, we need to introduce a truncation with controllable error. 
\begin{lem}[Truncating Bohr frequencies]\label{lem:truncating_bohr}
Consider the setting of~\autoref{lem:identify}. For any $\Omega'>0$, we have that
    \begin{align}
    \frac{\beta\sqrt{2\sigma\sqrt{2\pi}}}{2}\braket{ \vO,[\vA,\vH-\vH']}_{\vrho}
    &= \frac{1}{\sqrt{2\pi}}\int_{-\infty}^{\infty}\int_{-\infty}^{\infty}\tr\bigg[\vO^{\dagger}_{\vH}(t) \big( h_+(t')\vA_{\vH'}(t'+t)\vrho  - h_-(t')\vrho\vA_{\vH'}(t'+t)\big) \bigg] g_{\beta}(t)\rd t'\rd t\\
    &\quad+ \frac{\beta}{2}\int_{\labs{\omega'}\ge \Omega'} \braket{\vO,[\hat{\vA}_{\vH'}(\omega'),\vH-\vH']}_{\vrho} \rd \omega', 
\end{align}

where 
\begin{align}
    g_{\beta}(t) := \frac{2}{\beta}g(2t/\beta)&\lesssim \frac{1}{\beta}e^{-2\pi \labs{t}/\beta}\quad\text{and}\quad 
    \labs{h_+(t)},\labs{h_-(t)} \lesssim e^{-\sigma^2t^2} \frac{\sqrt{\sigma}}{\beta}e^{\beta \Omega'/2+\sigma^2\beta^2/4}\quad \text{for each}\quad t\in \BR.\label{eq:h_bounds}
\end{align}
\end{lem}

\begin{proof}    
To ease the notation, we distinguish two Hamiltonians $\vH,\vH'$ by the scalar variables $\omega$ and drop the subscripts $\vH,\vH'$  
\begin{align}
    \hat{\vA}_{\vH'}(\omega') &\equiv \hat{\vA}(\omega'), \quad \hat{\vA}_{\vH}(\omega) \equiv \hat{\vA}(\omega).
\end{align}
That is, whenever we use $\omega',$ we meant for operator Fourier transform with respect to the Hamiltonian $\vH'.$ We introduce a truncation frequency $\Omega'>0$ by~\autoref{lem:sumoverenergies}
\begin{align}
    c\vA = \int_{\labs{\omega'}\le \Omega'}\hat{\vA}_{\vH'}(\omega')\rd \omega' + \int_{\labs{\omega'}\ge \Omega'}\hat{\vA}_{\vH'}(\omega')\rd \omega' 
\end{align}
where $c = \sqrt{2\sigma\sqrt{2\pi}}.$ Then, we may rewrite
\begin{align}
    &\frac{\beta c}{2}\braket{ \vO,[\vA,\vH-\vH']}_{\vrho} \\
    &= \frac{1}{\sqrt{2\pi}} \int_{-\infty}^{\infty}\tr\L[\vO^{\dagger}_{\beta\vH/2}(t) \int_{\labs{\omega'}\le \Omega'} \L(  \sqrt{\vrho'}\hat{\vA}(\omega')_{\beta \vH'/2}(t)\sqrt{\vrho^{'-1}} \vrho - \vrho\sqrt{\vrho^{'-1}}\hat{\vA}(\omega')_{\beta \vH'/2}(t)\sqrt{\vrho^{'}} \R)\R] \rd \omega' g(t) \rd t\\
    &\quad+ \frac{\beta}{2}\int_{\labs{\omega'}\ge \Omega'} \braket{\vO,[\hat{\vA}_{\vH'}(\omega'),\vH-\vH']}_{\vrho} \rd \omega'. 
\end{align}

Now, we rewrite the first term in terms of the time average of Heisenberg dynamics, which will be manifestly efficient to implement. Since the operator Fourier transform and Heisenberg dynamics commute, we have that 
\begin{align}
\sqrt{\vrho'}(\hat{\vA}_{\vH'}(\omega'))_{\beta \vH'/2}(t)\sqrt{\vrho^{'-1}}  = \L(\sqrt{\vrho'}\hat{\vA}_{\vH'}(\omega')\sqrt{\vrho^{'-1}} \R)_{\beta \vH'/2}(t) \quad \text{for each}\quad \omega', t \in \BR.
\end{align}
Now, 
\begin{align}
\label{eq:Gibbsfrequencyintime}
    \int_{\labs{\omega'}\le \Omega'} \sqrt{\vrho'}\hat{\vA}_{\vH'}(\omega')\sqrt{\vrho^{'-1}}\rd \omega' &= \int_{\labs{\omega'}\le \Omega'} \vA_{\vH'}(\omega'-\sigma^2\beta) e^{-\beta \omega/2+\sigma^2\beta^2/4} \rd \omega' \\
&= \frac{1}{\sqrt{2\pi}} \int_{-\infty}^{\infty} \vA_{\vH'}(t')\int_{\labs{\omega'}\le \Omega'} \e^{-\ri (\omega'-\sigma^2\beta) t'} e^{-\beta \omega'/2+\sigma^2\beta^2/4}\rd \omega' f(t') \rd t' \\
& =: \frac{1}{\sqrt{2\pi}} \int_{-\infty}^{\infty}\vA_{\vH'}(t') h_+(t')\rd t'.
\end{align}
Similarly,
\begin{align}
\label{eq:Gibbsfrequencyintime1}
    \int_{\labs{\omega'}\le \Omega'} \sqrt{\vrho^{'-1}}\hat{\vA}_{\vH'}(\omega')\sqrt{\vrho'}\rd \omega' = \frac{1}{\sqrt{2\pi}} \int_{-\infty}^{\infty}\vA_{\vH'}(t') h_-(t')\rd t'.
\end{align}
In both cases above,
\begin{align}
    \labs{h_+(t)},\labs{h_-(t)} \lesssim \labs{f(t)} \frac{1}{\beta}e^{\beta \Omega'/2+\sigma^2\beta^2/4}\lesssim e^{-\sigma^2t^2} \frac{\sqrt{\sigma}}{\beta}e^{\beta \Omega'/2+\sigma^2\beta^2/4}.
\end{align}

Absorb the factors of $\beta/2$ in $\vO_{\beta \vH/2}(t)$ and $\vA_{\beta \vH'/2}(t')$ by rescaling $g_{\beta}(t) = \frac{2}{\beta}g(2t/\beta)$ and merge the two Heisenberg dynamics $t,t'$ to conclude the proof.
\end{proof}

The truncation error can be bounded as follows.
  
\begin{lem}[Decay of high-frequency part]
\label{lem:high_freq}
Consider a Hamiltonian $\vH'= \sum_{\gamma}\vh_{\gamma}$ with interaction degree $d$ (as in~\autoref{sec:Ham}) and $\vG = \sum_{\gamma} \vg_{\gamma}$ with the same interaction graph as $\vH'$ and $\|\vh_{\gamma}\|,\|\vg_{\gamma}\|\leq 1$. In the setting of~\autoref{lem:identify} and for single site $\vA,$
\begin{align}
\labs{\int_{\labs{\omega'}\ge \Omega'} \braket{\vO,[\hat{\vA}_{\vH'}(\omega'),\vG]}_{\vrho} \rd \omega'} \lesssim \frac{d^{4+16e^2d^4/\sigma^2}}{\sqrt{\sigma}} e^{-\Omega'/4d + \sigma^2/16d^2} \norm{\vO}\norm{\vA}.
\end{align}    
\end{lem}
Here, we need to carefully exploit the locality of $\hat{\vA}_{\vH'}(\omega')$ and $\vG$, by expanding the imaginary and real time evolution; the quasi-locality contributes to the factor of $d^{4+16e^2d^4/\sigma^2}$. The proof is routine, see~\autoref{sec:proof_high_freq}.
\begin{rmk}
    The RHS grows as $e^{\CO(\beta^2)}$ when we set the uncertainty $\sigma = 1/\beta.$ This dependence on $\beta$ can be further improved for lattice Hamiltonians, but we do not pursue it here as we will lose factors of $e^{\beta^D}$ elsewhere which dominates the learning sample complexity.
\end{rmk}
\subsection{Local commutators are faithful}\label{sec:KMSfaithful}
Recall the inequality
\begin{align}
&\norm{[\vA,\vH-\vH']}^2_{\vrho} \le 2d \sup_{\vO = [\vA,\vP_{\gamma}]}\big\lvert \braket{ \vO,[\vA,\vH-\vH']}_{\vrho}\big\rvert
\end{align}
As long as the RHS is small, the following lemma guarantees that $[\vA,\vH-\vH']$ is also small. This is due to the following lemma. 
\begin{lem}[KMS is locally faithful]
\label{lem:KMSvariance} Assuming $\beta \geq 1/4d$.
For a (not necessarily Hermitian) operator $\vB$ and Gibbs state $\vrho = \e^{-\beta \vH}/\tr(\e^{-\beta \vH})$ of an interaction degree-$d$ Hamiltonian $\vH$, 
it holds that $$\|\vB\|_{{\vtau}} \leq \e^{ 80\beta |\mathrm{supp}(\vB)| + 16d\beta \log 2d\beta }\|
\vB\|_{\vrho},$$ where $\vtau$ is the maximally mixed state.
\end{lem}
The KMS norm $\norm{\vB}^2_{\vrho}$ on finite-temperature Gibbs states provides an upper bound on the variance of local operators on the maximally mixed state. Since $\vrho$ is invertible, $\norm{\vB}^2_{\vrho}=0$ already implies $\vB=0.$ The goal is to obtain quantitative bounds for local operators $\vB$ that are independent of the system size. The argument exploits the regularization trick (\autoref{sec:regularize}).

\begin{proof}[Proof of~\autoref{lem:KMSvariance}]

We will establish two claims to obtain a self-bounding argument.

\textbf{Claim 1.} We can relate $\|\vB\|_{\vtau} $ to the $\vrho$-weighted KMS norm a rotated version of $\vB$. In particular, there exist unitaries $\vU, \vV$ supported on $\text{supp}(\vB)$ such that 
\begin{align}
    \|\vB\|_{\vtau} \leq 2^{2|\text{supp}(\vB)|} \|\vU^\dagger \vB \vV\|_{\vrho}. \label{eq:rotated}
\end{align}

Denote support of $\vB$ by $B= \text{supp}(\vB)$. Consider the Haar average over the region $B$,
\begin{equation}
    \mathbb{E}_{\vU, \vV} \tr( \vB^\dagger \vU\sqrt{\vrho} \vU^\dagger \vB \vV \sqrt{\vrho} \vV^\dagger) =  \tr((\tr_B(\sqrt{\vrho}))^2) \frac
    {\|\vB\|_{\vtau}^2}{2^{|B|}},
\end{equation}
where we used that $\mathbb{E}_{\vU} \vU \sqrt{\vrho} \vU^\dagger = \tr_B(\sqrt{\vrho}) \otimes \frac{\vI_B}{2^{|B|}}$.
Next we lower bound the expression $\tr((\tr_B(\sqrt{\vrho}))^2)$. We have
\begin{align}
    \sqrt{\vrho} &\leq 2^{|B|} \vI_B \otimes \tr_B(\sqrt{\vrho}) \\
    \Longrightarrow \ 1=\tr(\vrho) &\leq 2^{3|B|}  \tr((\tr_B(\sqrt{\vrho}))^2),
\end{align}
where the last line used $\tr(\vC^2) \leq \tr(\vD^2)$ when $0\leq \vC\leq \vD$ since $\tr(\vD^2-\vC^2)=\tr((\vD-\vC)(\vD+\vC))$. Combine the above to obtain Eq.~\eqref{eq:rotated}. 

\textbf{Claim 2.} The KMS inner product is `protected' from local rotations (analogous to \cite[Proposition 10]{anshu2020sample}). For unitaries $\vU, \vV$ supported on a region $B$, it holds that
\begin{align}
     \frac{\|\vU^\dagger \vB \vV\|_{\vrho}}{\|\vB\|}\leq
   (1 + e^{1/4d\beta}) \L( \frac{\|\vB\|_{\vrho}}{\|\vB\|} \R)^{1/2}+   (2+8^{|B|}d\beta)8^{|B|}d\beta \left(\frac{\|\vB\|_{\vrho}}{\|\vB\|}\right)^{1/8d\beta}.
   \label{eq:rotation-protected}
\end{align}
Consider the decomposition from~\autoref{lem:sumoverenergies}
\begin{align}
    \vU= \frac{1}{\sqrt{\sigma2\sqrt{2\pi}}} \left(\int_{|\omega|\leq \Delta}\hat{\vU}(\omega)d\omega + \int_{|\omega|\geq \Delta}\hat{\vU}(\omega)d\omega \right) := \vU_{\leq \Delta} + \vU_{\geq \Delta}
\end{align}
for tunable $\sigma>0$ and $\Delta>0.$  
Using \autoref{cor:norm_decay} and~\autoref{lem:convergence_imaginary} with $|\beta_0| =1/4d$, we can bound
\begin{align}
    \|\vU_{\geq \Delta}\| \leq \frac{1}{\sqrt{\sigma2\sqrt{2\pi}}} \int_{|\omega|\geq \Delta} \|\hat{\vU}(\omega) \| d\omega \leq \frac{1}{\sigma\sqrt{2\pi}} \int_{|\omega|\geq \Delta} 4^{|B|}e^{-|\omega|/4d + \sigma^2/16d^2} d\omega \leq \frac{4d}{\sigma} 4^{|B|} e^{\sigma^2/16d^2-\Delta/4d},
\end{align}
where the factor $4^{|B|}$ comes from taking the Pauli decomposition $\vU= \sum_{\vP} a_{\vP} \vP$, applying~\autoref{lem:convergence_imaginary} to each Pauli $\|e^{\beta_0\vH} \vP e^{-\beta_0\vH}\| \leq 2^{|B|}$, and $\sum_{\vP} |a_{\vP}|\leq 2^{|B|}$. At this stage, let us set $\sigma=1/\beta$, which means that $ \|\vU_{\geq \Delta}\|\leq d\beta 8^{|B|}e^{1/16d^2\beta^2-\Delta/4d}$ and $\|\vU_{\leq \Delta}\| \leq \|\vU\| + \|\vU_{\geq \Delta}\| \leq 1+ d\beta 8^{|B|} e^{1/16d^2\beta^2-\Delta/4d}$. We take a similar decomposition of $\vV=\vV_{\leq \Delta} + \vV_{\geq \Delta}$.

With the above decompositions and using that $|\langle \vP,\vQ \rangle_{\vrho}| \leq \|\vP\| \|\vQ\|$ and the triangle inequality we can bound
\begin{align}
    \|\vU^\dagger \vB \vV\|_{\vrho}
    \leq \|\vU_{\leq \Delta}^\dagger \vB \vV_{\leq \Delta}\|_{\vrho} + (2+ d\beta 8^{|B|} e^{1/16d^2\beta^2-\Delta/4d})d\beta 8^{|B|} e^{1/16d^2\beta^2-\Delta/4d} \|\vB\|.
\end{align}
Next, we bound the first term from the above expression
\begin{align}
    & \tr\L(\vV_{\leq \Delta}^\dagger \vB^\dagger \vU_{\leq \Delta} \sqrt{\vrho} \vU_{\leq \Delta}^\dagger \vB \vV_{\leq \Delta} \sqrt{\vrho}\R) \\
    &= \tr\L((\vrho^{1/4} \vV^\dagger_{\leq \Delta} \vrho^{-1/4})  (\vrho^{1/4} \vB^\dagger \vrho^{1/4}) (\vrho^{-1/4} \vU_{\leq \Delta} \vrho^{1/4})  (\vrho^{1/4} \vU^\dagger_{\leq \Delta} \vrho^{-1/4}) (\vrho^{1/4} \vB \vrho^{1/4}) (\vrho^{-1/4} \vV_{\leq \Delta} \vrho^{1/4}) \R) \\
    &\leq    \L(\|\vrho^{-1/4} \vV_{\leq \Delta} \vrho^{1/4} \| \|\vrho^{-1/4} \vU_{\leq \Delta} \vrho^{1/4} \| \|\vB\|_{\vrho}\R)^2 \tag*{(Holder's inequality)}\\
    &\leq  \L(\frac{1}{\sigma2\sqrt{2\pi}} \int_{|\omega|\leq \Delta} d\omega \|\vrho^{-1/4} \hat{\vV}(\omega) \vrho^{1/4} \| \int_{|\omega|\leq \Delta} d\omega \|\vrho^{-1/4} \hat{\vU}(\omega) \vrho^{1/4} \|  \|\vB\|_{\vrho}\R)^2 \\
    &\leq  \left(\frac{e^{\sigma^2 \beta^2/16}}{\sigma 2\pi\sqrt{2}} \int_{|\omega|\leq \Delta} d\omega e^{\beta \omega/4}
 \right)^4 \|\vB\|_{\vrho}^2 \tag*{(\autoref{lem:bounds_imaginary_conjugation})} \\
 &\leq  e^{\beta\Delta}
  \|\vB\|_{\vrho}^2.
\end{align}
Set $\Delta= \max ( \frac{1}{\beta}\log\frac{\|\vB\|}{ \|\vB\|_{\vrho}} ,  \frac{1}{2d\beta^2})$ and combine the above bounds to obtain Eq.~\eqref{eq:rotation-protected}.

Finally, the two claims~\eqref{eq:rotated},~\eqref{eq:rotation-protected} imply that
\begin{align}
    \frac{\|\vB\|_{\vtau}}{\|\vB\|} \leq  \L((1 + e^{1/4d \beta}) \L( \frac{\|\vB\|_{\vrho}}{\|\vB\|} \R)^{1/2}+   (2+8^{|B|}d\beta)8^{|B|}d\beta \left(\frac{\|\vB\|_{\vrho}}{\|\vB\|}\right)^{1/8d \beta}\R) 2^{2|B|}.
\end{align}
Since $\|\vB\|\leq 2^{|B|} \|\vB\|_{\vtau}$, assuming $\beta\geq 1/4d$ we obtain that
$$\frac{\|\vB\|_{\vrho}}{\|\vB\|} \geq e^{-80\beta|\text{supp}(\vB)| - 16 d\beta \log 2d\beta}.$$
Further, since $\|
\vB\|_{\vtau} \leq \|\vB\|$, the lemma follows.
\end{proof}

Once we know that $\|[\vA,\vH-\vH']\|_{\vtau}$ is small for each local Pauli $\vA$ on a qubit $i$, it is clear that $\vH,\vH'$ are close to each other near $i$, by a direct computation.

\begin{lem}[Locally good coefficients]\label{lem:localcloseness}
Consider Hamiltonians $\vH = \sum_{\gamma \in \Gamma}h_{\gamma} \vP_{\gamma}$ and $\vH' = \sum_{\gamma \in \Gamma}h'_{\gamma} \vP_{\gamma}$ for distinct Pauli strings $\vP_{\gamma}$. If $\|[\vA,\vH-\vH']\|_{\vtau} \leq \epsilon$ for each $\vA$ from the set of single-qubit $\{\vA^a\} = \{\vX_i,\vY_i,\vZ_i\}$ on a particular qubit $i$, then, 
\begin{align}
    |h_\gamma -h'_\gamma| \leq \epsilon\quad \text{for each}\quad \vP_{\gamma} \quad \text{acting on qubit}\quad i.
\end{align}
\end{lem}
\begin{proof} It holds that
    \begin{align}
    \sum_{a=1,2,3}\norm{[\vA^a,\vH-\vH']}^2_{\tau} &= \sum_{a=1,2,3} \frac{1}{2^n}\tr[ [\vA^a,\vH-\vH'][\vA^a,\vH-\vH']^{\dagger}]\\
    &= \sum_{a=1,2,3} \frac{1}{2^n}\tr[ [\vA^a,[\vA^a,\vH-\vH']](\vH-\vH')]
    = 8\sum_{\gamma\sim i}(h_{\gamma}-h'_{\gamma})^2
\end{align}
using that the double commutator pick out Pauli strings that acts on the qubit $$\sum_{a=1,2,3}[\vA^a,[\vA^a,\vH-\vH']] = 8\sum_{\gamma\sim i} (h_{\gamma}-h'_{\gamma})\vP_{\gamma}$$ and that Paulis are orthonormal $\frac{1}{2^n}\tr[\vP_{\gamma}\vP^{\dagger}_{\gamma'}]=\delta_{\gamma,\gamma'}.$ This shows that $8\sum_{\gamma \sim i} (h_\gamma-h_\gamma')^2 \leq  3\epsilon^2$ and hence $|h_\gamma-h_\gamma'| \leq \epsilon$ for each $\gamma$ in the sum.
\end{proof}

\section{The learning protocol}
\label{sec:learning_protocol}
In this section, we turn the identifiability observable into a local learning algorithm. For concreteness, we assume that the Hamiltonian terms are each a distinct, known Pauli operator $\vP_{\gamma}$,
\begin{align}
    \vH = \sum_{\gamma\in \Gamma} \vh_{\gamma} = \sum_{\gamma\in \Gamma} h_{\gamma} \vP_{\gamma}
\end{align}
and we wish to learn the unknown parameter $h_{\gamma}\in [-1,1]$ for each $\gamma\in \Gamma.$ For the entire~\autoref{sec:learning_protocol}, we will also set the uncertainty in operator Fourier transform to be
\begin{align}
    \sigma = \frac{1}{\beta},
\end{align}
which appears sufficient. To simplify the computation, we also often assume that 
 \begin{align}
 \beta \ge \frac{1}{d}.
 \end{align}
If the input $\beta$ is too small, we rescale the Hamiltonian term $\vH \rightarrow c\vH$ so that the above hold; this will save us from repeating similar arguments for the small $\beta$ regime.

Inspired by the identifiability equation, we begin by defining an observable $Q$ that will play a key role in the learning protocol. The quantitative guarantees will depend on the locality and stability of this observable.

\subsection{Robustness of the identifiability observable \texorpdfstring{$Q$}{q}}

Crucial to our protocol is a quantity capturing the identifiability of the Hamiltonian via (quasi-)\textit{local} measurements. Recall the identifiability observable from the introduction:
\begin{align}
    Q(\vO,\vG,\vA,\vK) = \frac{1}{\sqrt{2\pi}}\int_{-\infty}^{\infty}\int_{-\infty}^{\infty}\tr\bigg[\vO^{\dagger}_{\vG}(t) \big( h_+(t')\vA_{\vK}(t'+t)\vrho  - h_-(t')\vrho\vA_{\vK}(t'+t)\big) \bigg] g_{\beta}(t)\rd t'\rd t
\end{align}
where implicitly $\vrho \propto e^{-\beta \vH}$, and the Hamiltonian $\vG$ and $\vK$ may not apriori be the same as $\vH.$ Indeed, to make use of~\autoref{lem:identify}, we do not a priori know the ground truth $\vH$, so we would have to also test against yet another Hamiltonian $\vG\ne\vH.$ We will see that, if the test Hamiltonian $\vK$ is close to the true Hamiltonian $\vH$, then the expression is small; conversely, if we are far from the true Hamiltonian, then the expression is large. 
We first derive some continuity properties of $Q$ in the presence of distant perturbations. In particular, we will often consider truncations in a similar fashion as $\vH_{\ell},$  
\begin{align}
   \vG = \sum_{\gamma\in \Gamma} \vg_{\gamma}, \quad \norm{\vg_\gamma}\le 1\quad\text{with}\quad  \vG_{\ell} &= \sum_{\gamma:\text{dist}(\gamma,\vO)<\ell-1} \vg_{\gamma},\\
    \vK = \sum_{\gamma\in \Gamma} \vk_{\gamma}, \quad \norm{\vk_\gamma}\le 1\quad\text{with}\quad \vK_{\ell} &= \sum_{\gamma:\text{dist}(\gamma,\vA)<\ell-1} \vk_{\gamma}.
\end{align}
and similarly for $\vG',\vK'$ and $\vG'_{\ell}, \vK'_{\ell}$. Right now, the identifiability observable $Q$ also depends on arbitrary operators $\vA$ and $\vO$, but we will always apply to single-site Paulis $\vA=\{\vX_i,\vY_i,\vZ_i\}$ and $\vO \propto [\vA,\vP_{\gamma}].$

\begin{lem}[Robustness of $Q$] \label{lem:stability_Q}
     Consider Hamiltonians $\vG, \vG', \vK, \vK'$ with the same interaction graph as in~\autoref{sec:Ham} such that $\norm{\vg_{\gamma}},\norm{\vg'_{\gamma}},\norm{\vk_{\gamma}},\norm{\vk'_{\gamma}}\le 1$, $\norm{\vg_{\gamma}-\vg'_{\gamma}},\norm{\vk_{\gamma}-\vk'_{\gamma}}\le \kappa$,
     and operators $\vA,\vO$ such that $\norm{\vA},\norm{\vO}\le 1.$ Assume that $\beta\ge 1/d$. Then, \\
     (A) The truncation error is bounded by
    \begin{align}
        \labs{Q(\vO,\vG,\vA,\vK) - Q(\vO,\vG_{\ell},\vA,\vK_{\ell})} \lesssim  \frac{e^{\beta \Omega'/2}}{\sqrt{\beta}} (e^{-\ell^2/16e^4d^2\beta^2}+e^{-\pi\ell/e^2d\beta})(|O|+|A|).
    \end{align}
(B) For extensive perturbations in which $\vg_\gamma=\vg'_\gamma$, $\vk_\gamma=\vk'_\gamma$ for all $\gamma$ at distance within $\ell_0$ from $\vO,\vA$,
    \begin{align}
        \labs{Q(\vO,\vG,\vA,\vK) - Q(\vO,\vG',\vA,\vK')} \lesssim \kappa \frac{e^{\beta \Omega'/2}}{\sqrt{\beta}} \sum_{\ell=\ell_0}^{\infty} (S(\ell,\vA)+S(\ell,\vO)) (\beta+\frac{\ell}{d})(e^{-\ell^2/16e^4d^2\beta^2}+e^{-\pi\ell/2e^2d\beta}).
    \end{align}
    (C) 
    For perturbation within a radius $\ell_0$, i.e., $\vg_\gamma=\vg'_\gamma$, $\vk_\gamma=\vk'_\gamma$ for all $\gamma$ at distance $\ell_0$ or larger from $\vO,\vA$,
        \begin{align} \labs{Q(\vO,\vG,\vA,\vK) - Q(\vO,\vG',\vA,\vK')} \lesssim \frac{\kappa \sqrt{\beta}}{d} e^{\beta \Omega'/2} ( V(\ell_0,\vO)+V(\ell_0,\vA)).
        \end{align}
\end{lem}

See Appendix~\ref{sec:proof_stability} for the proof using routine Lieb-Robinson arguments.
\begin{rmk}
To learn the Hamiltonian to a high precision, we will exploit the fact that the surface area $S(\ell)$ grows polynomially with the distance $\ell$. 
\end{rmk}

\subsection{ Identifiability of test Hamiltonian: existence and uniqueness}
\label{sec:exist_unique} 
Here, we derive operational properties of the identifiability observable that will help us interpret the experimental values of $Q$. Essentially, $Q$ gives a unique way to identify when a local guess is approximately correct. Indeed, as a consistency check, inserting the ground truth Hamiltonian $\vH$ always gives a vanishing $Q$ for any $\vG$.

\begin{lem}[Existence of a global, perfect guess]\label{lem:existence-global}
Recall the ground truth Gibbs state $\vrho = e^{-\beta \vH}/\tr(e^{-\beta \vH}).$ Then, the identitifiability observable vanishes exactly
    \begin{align}
    Q(\vO,\vG,\vA,\vH) = 0 \quad \text{for each Hamiltonian}\quad \vG \quad \text{and} \quad \vO,\vA.
\end{align}
\end{lem}
\begin{proof}
   As in the proof of~\autoref{lem:truncating_bohr}, for $\vrho' \propto e^{-\beta \vK}$, we have the exact identity
\begin{align*}
    Q(\vO,\vG,\vA,\vK) = \frac{1}{\sqrt{2\pi}} \int_{-\infty}^{\infty}\tr\L[\vO^{\dagger}_{\vG}(t) \int_{-\Omega'}^{\Omega'} \L(  \sqrt{\vrho'}\hat{\vA}_{\vK}(\omega')_{\vK}(t)\sqrt{\vrho^{'-1}} \vrho - \vrho\sqrt{\vrho^{'-1}}\hat{\vA}_{\vK}(\omega')_{\vK}(t)\sqrt{\vrho^{'}} \R)\R] \rd \omega' g_{\beta}(t) \rd t.
\end{align*}
When $\vK = \vH$, we have that $\vrho' = \vrho$ and for every $t,\omega',$
\begin{align}
    \sqrt{\vrho'}\hat{\vA}_{\vK}(\omega')_{\vK}(t)\sqrt{\vrho^{'-1}} \vrho - \vrho\sqrt{\vrho^{'-1}}\hat{\vA}(\omega')_{\vK}(t)\sqrt{\vrho^{'}}  =  \sqrt{\vrho}\hat{\vA}_{\vH}(\omega')_{\vH}(t)\sqrt{\vrho^{-1}} \vrho - \vrho\sqrt{\vrho^{-1}}\hat{\vA}(\omega')_{\vH}(t)\sqrt{\vrho} =0,
\end{align}
as advertised.
\end{proof}

In our algorithm, we will only make local guesses in a search radius $\ell$, and we will have to discretize the set of parameter guesses $k_\gamma$ by introducing an epsilon net. For each coefficient labeled by $\gamma,$ consider the set of discrete points
\begin{align}
    N^{\kappa} \subset [-1,1]\quad \text{such that}\quad x \in [-1,1] \implies \labs{x- N^{\kappa}} \le \kappa.
\end{align}
Of course, such a set can be chosen to have cardinality $\labs{N^{\kappa}} = \lceil 2/\kappa\rceil.$ We will denote Hamiltonians whose coefficients are chosen from the epsilon net with a prime, such as $\vG',\vK'$ and $\vG'_{\ell}, \vK'_{\ell}$. 

Exploiting the stability of $Q$, the following lemma states that a locally correct guess must also behave like the ground truth. The larger the local patch, the better $Q$ is.

\begin{lem}[Existence of a good localized Hamiltonian on the epsilon net]
    \label{lem:exsitence} Assume that $\beta\ge 1/d$. For every $\vA,\vO$ such that $\norm{\vA},\norm{\vO}\le 1$, there exists a $\vK'_\ell = \sum_{\gamma} k'_{\gamma}\vP_{\gamma}$ with $k'_{\gamma} \in N_\kappa$ such that for every $\vG'_\ell = \sum_{\gamma} g'_{\gamma}\vP_{\gamma}$ with $g'_{\gamma} \in N_\kappa$, the identifiability observable satisfies
    \begin{align}
        \labs{Q(\vO,\vG'_\ell,\vA,\vK'_\ell)} \lesssim \frac{e^{\beta \Omega'/2}}{\sqrt{\beta}} \L( e^{-\pi\ell/e^2d\beta} + \kappa \beta(V(\ell,\vA)+V(\ell,\vO))  \R).\end{align}
\end{lem}

\begin{proof}
The idea is to take the ground truth Hamiltonian $\vK = \vH$, localize to $\vK_{\ell}$, and then round it on the epsilon net. 
By the stability of $Q$ (Item (A) of \autoref{lem:stability_Q}), we may truncate $\vK \rightarrow \vK_{\ell}$
\begin{align}
    \labs{Q(\vO,\vG,\vA,\vK_{\ell})} \lesssim \frac{1}{\sqrt{\beta}}e^{\beta \Omega'/2} de^{-\pi\ell/e^2d\beta}
\end{align}
for any $\vG,$ and particularly $\vG'_{\ell}$ from the epsilon net.

Next, we round the Hamiltonian $\vH_{\ell}$ to the epsilon net. By Item (C) of \autoref{lem:stability_Q}, 
\begin{align}
\labs{Q(\vO,\vG'_\ell,\vA,\vK_\ell)-Q(\vO,\vG'_\ell,\vA,\vK'_\ell)} \lesssim
\frac{\kappa \sqrt{\beta}}{d} e^{\beta \Omega'/2}(V(\ell,\vA)+V(\ell,\vO)).    
\end{align}
Collect the errors to conclude the proof.
\end{proof}

The remarkable feature of the identifiability equation is that a guess $\vH_{\ell} $ that achieves good values of $Q$ must simply be locally correct.

\begin{lem}[Uniqueness of good local guesses]\label{lem:unique}
Assume that $\beta\ge 1/d$. For every $\vA,\vO$ such that $\norm{\vA},\norm{\vO}\le 1$, 
suppose there is an $\vH' = \sum_{\gamma} h'_{\gamma}\vP_{\gamma}$ such that $\labs{Q(\vO,\vG'_\ell,\vA,\vH')}\le \epsilon$ for every $\vG'_\ell = \sum_{\gamma} g'_{\gamma}\vP_{\gamma}$ with $g'_{\gamma} \in N_\kappa$. 
Then,
\begin{align}
    \labs{\braket{ \vO,[\vA,\vH-\vH']}_{\vrho}} &\lesssim \frac{e^{\beta \Omega'/2}}{\beta}  de^{-\pi\ell/e^2d\beta}+ \beta e^{-\Omega'/4d} d^{4+16e^2d^4\beta^2}  + \frac{\kappa}{d} e^{\beta \Omega'/2} ( V(\ell,\vO)+V(\ell,\vA)) + \frac{\epsilon}{\sqrt{\beta}}.
\end{align} 
\end{lem}
\begin{proof}
    By~\autoref{lem:truncating_bohr},
\begin{align}
    \frac{\beta\sqrt{2\sigma\sqrt{2\pi}}}{2}\braket{ \vO,[\vA,\vH-\vH']}_{\vrho}= Q(\vO,\vH,\vA,\vH')+\frac{\beta}{2}\int_{\labs{\omega'}\ge \Omega'} \braket{\vO,[\hat{\vA}_{\vH'}(\omega'),\vH-\vH']}_{\vrho} \rd \omega'.
\end{align}
By the stability of $Q$ (adapting Item (A) of~\autoref{lem:stability_Q} for only changing the $\vH$ argument), and setting $\sigma = 1/\beta$ for the operator Fourier transform,
\begin{align}
    \labs{Q(\vO,\vH,\vA,\vH')} &\le \labs{Q(\vO,\vH,\vA,\vH')-Q(\vO,\vH_{\ell},\vA,\vH')}+ \labs{Q(\vO,\vH_{\ell},\vA,\vH')}\\
    &\lesssim \frac{e^{\beta \Omega'/2}}{\sqrt{\beta}} de^{-\pi\ell/e^2d\beta} +\labs{Q(\vO,\vH_{\ell},\vA,\vH')}. 
\end{align}
Next, we use Item (C) of~\autoref{lem:stability_Q} to round $\vH_\ell$ ot the epsilon net. There exists a $\vG'_\ell$ such that
\begin{align}
    \labs{Q(\vO,\vH_{\ell},\vA,\vH') - Q(\vO,\vG'_{\ell},\vA,\vH')} \leq \frac{\kappa \sqrt{\beta}}{d} e^{\beta \Omega'/2} ( V(\ell,\vO)+V(\ell,\vA)).
\end{align}
Recall the bound on the high-frequency part (\autoref{lem:high_freq}) to conclude the proof.

\end{proof}

\subsection{Measuring the identifiability observables}
This section summarizes the cost of measuring the identifiability observables. Due to locality, arguments are routine (see Appendix~\ref{sec:measurement_proof}). Here, $\CO(\cdot), \mathrm{poly}(\cdot)$ suppress dependence on the interaction degree $d$ and locality $q$.
\begin{lem}[Measuring a single $Q$]\label{lem:measure-single-Q} 
On a bounded degree interaction graph (\autoref{sec:Ham}), the observable $Q(\vO, \vG'_{\ell},\vA,\vK'_{\ell})$, where the coefficients of $\vG'_\ell, \vK'_\ell$ are taken from the net $N_\kappa$, can be measured to precision $\epsilon$ with probability $1-\delta$ using 
\begin{center}
$\CO\L(\frac{e^{\beta \Omega'} \|\vA\|^2 \|\vO\|^2}{\beta \epsilon^2} \log (1/\delta)\R)$ copies of $\vrho$, and\\
$\mathrm{poly}\L(\beta, V(\ell),\log(1/\kappa \epsilon), \frac{e^{\beta \Omega'/2}}{\sqrt{\beta}}\|\vA\|\|\vO\|\R)\frac{\log(1/\delta)}{\epsilon^2}$ elementary quantum gates 
\end{center}
acting on the neighbourhoods $V(\ell,\vO)\cup V(\ell,\vA)$. 
\end{lem}

Since the measurements are local, we may measure multiple identifiability observables $Q$ in parallel, given some understanding of how the observables overlap with each other. The sample, time complexity, and performance guarantee directly follow from the above lemma.

\begin{alg}[Measuring all $Q$]\label{alg:measureQ} 
On a bounded degree interaction graph (\autoref{sec:Ham}),
consider a set $S$ of identifiability observables $Q(\vO,\vG'_\ell,\vA,\vK'_\ell)$ for single-site Pauli $\vA \in \{\vX_i,\vY_i,\vZ_i\}, i \in \Lambda$, nonzero $\vO = [\vA,\vP_\gamma]$, $\gamma \in \Gamma$, and $\vG'_\ell$, $\vK'_\ell$ supported on $V(\ell,\vO)$, $V(\ell,\vA)$, respectively. Let $\chi-1$ be the maximum number of $Q$'s that overlap with a single $Q$. Let $\epsilon$ be the precision parameter and $p_\mathrm{fail}$ be the desired probability of failure.
\begin{enumerate}
    \item (Partition into non-overlapping subsets) Partition $S$ into subsets $S_1,\cdots,S_\chi$ such that within each subset $S_i$, the identifiability observables $Q$ are non-overlapping.
    \item (Parallel measurements) For each $S_i$, perform the algorithm from~\autoref{lem:measure-single-Q} with precision $\epsilon$ and $\delta = p_\mathrm{fail}/|S|$ in parallel and output the estimate $Q_\mathrm{exp}$ for each $Q$ in the subset. 
\end{enumerate}

\vspace{0.1in}

\noindent\textbf{Complexity.} The algorithm uses $\CO(\chi \frac{e^{\beta \Omega'}}{\beta \epsilon^2} \log (|S|/p_\mathrm{fail}))$ copies of $\vrho$ and $$\mathrm{poly}\L(\beta, V(\ell),\log(\frac{1}{\kappa \epsilon}), \frac{e^{\beta \Omega'/2}}{\sqrt{\beta}}\R) \frac{|S|\log (|S|/p_\mathrm{fail})}{\epsilon^2}$$ elementary quantum gates and classical processing time.

\vspace{0.1in}

\noindent\textbf{Guarantee.} With probability $1-p_\mathrm{fail}$, it holds for each $Q(\vO,\vG'_\ell,\vA,\vK'_\ell) \in S$ that the corresponding estimate $Q_\mathrm{exp}(\vO,\vG'_\ell,\vA,\vK'_\ell)$ returned in step 2 satisfies
\begin{align}
\labs{Q_\mathrm{exp}(\vO,\vG'_\ell,\vA,\vK'_\ell) - Q(\vO,\vG'_\ell,\vA,\vK'_\ell)}\le \epsilon.
\end{align}
\end{alg}

\subsection{A simple local learning algorithm for Hamiltonians with any connectivity}
\label{sec:local_anygraph}
We are now ready to give a local learning algorithm for quantum Gibbs states with any interaction graph with a bounded interaction degree $d$. In~\autoref{sec:local_anygraph}, $\CO(\cdot)$ suppresses the dependence on the geometry (degree bound $d$ and locality $q$). We introduce the suitable absolute constants $c_1,c_2,c_3$ so that the error analysis is strictly controlled by  `$\leq$' in later calculations (instead of $\lesssim$).

\begin{cond}\label{cond:param-simple}
Assume $\beta \ge 1/d$.
In~\autoref{sec:local_anygraph}, we set the following parameters. For a target precision $\epsilon$, an auxilliary parameter $\alpha = 2de^{200(d+q)\beta \log d\beta}$, and tunable absolute constants $c_1 \leq c_2 \leq c_3$, set
    \begin{itemize}
    \item Frequency truncation: $e^{\Omega'/4d} = c_1 \cdot 5\beta  d^{4+16e^2d^4\beta^2}\alpha/\epsilon^2 \Longrightarrow \Omega' =\CO( \beta^2 +  \log 1/\epsilon )$.
    \item Search truncation radius: $\ell= c_2 \cdot 10d\beta (\beta\Omega' + \log(5\alpha/\beta\epsilon^2))  \Longrightarrow \ell =\CO(\beta^4 + \beta  \log 1/\epsilon) $.
    \item Epsilon net precision: $\kappa=\frac{\epsilon^2}{c_3 \cdot 40\alpha} e^{-\beta \Omega'/2}{\sqrt{\beta}}d^{-\ell-3} \Longrightarrow \kappa = \epsilon^{2+\CO(\beta)} 2^{-\CO(\beta^4)}$.
\end{itemize}

\end{cond}

The algorithm is very straightforward: for each local term, we search over the local neighbourhood and return an assignment of coefficients that minimizes the identifiability observable (against any local test Hamiltonian $\vG$).

\begin{alg}[Learning each local term locally]
\label{alg:constant_err} Consider the Hamiltonian with a bounded degree interaction graph (\autoref{sec:Ham}) and target an error budget $\epsilon$ for each $h_{\gamma}$.
    \begin{enumerate}
        \item (Measure in parallel) 
        Perform experiments (\autoref{alg:measureQ}) to measure all observables 
        \begin{align}
        Q_{exp}([\vA,\vP_{\gamma}],\vG'_{\ell},\vA,\vK'_{\ell}) \quad \text{over inputs}\quad i\in \Lambda, \quad \vA\in\{\vX_i,\vY_i,\vZ_i\}, \quad \gamma: \gamma \sim i, \quad \vG'_{\ell},\vK'_{\ell}\in N_\kappa.
    \end{align}
    to precision $ \epsilon^2 \sqrt{\beta}/20\alpha$ and failure probability $1-p_\mathrm{fail}$.
    \item (Identify local terms) For each $i \in \Lambda$:
    \begin{itemize}
        \item  Identify the parameters of the Hamiltonian $\vK'_{\ell}$ which attains the weakest $Q$ for all $\vA,\gamma,\vG'_{\ell}$
    \begin{align}
        \min_{\vK'_{\ell}} \max_{\vA,\gamma,\vG'_{\ell}} \labs{Q_{exp}([\vA,\vP_{\gamma}],\vG'_{\ell},\vA,\vK'_{\ell})}.
    \end{align}
    \item     Record the terms in $\vK'_{\gamma}$ that acts on qubit $i$, and set  
    \begin{align}
        h'_{\gamma} \leftarrow k'_{\gamma}.
    \end{align}
    \end{itemize}
    \item Return the collection of coefficients $\{h'_{\gamma}\}_{\gamma \in \Gamma}.$
    \end{enumerate}
\end{alg}
\begin{rmk}
    The same coefficient $h_{\gamma}$ may be updated multiple times as we sweep through various sites $i$ near $\gamma$. In fact, any such $h'_{\gamma}$ is guaranteed to be close to the ground truth, and we merely need to return any one of them. We are also throwing away large chunks of $\vK'_{\ell}$ which do not act on a given site $i$. 
    \end{rmk}

\begin{thm}[Learning quantum Gibbs states locally - Thm \ref{thm:graphthm}]\label{thm:technical_graph}
Consider Gibbs state $\vrho_{\beta}$ for a Hamiltonian $\vH$ with constant locality $q$ and a bounded interaction degree $d$ (\autoref{sec:Ham}) at inverse temperature $\beta$. With the parameters from~\autoref{cond:param-simple},
    \autoref{alg:constant_err} learns an approximation $\vH'$ to the ground truth $\vH$ such that 
    \begin{align}
        \labs{h_{\gamma}-h'_{\gamma}} \le \epsilon \quad \text{for all}\quad \gamma \in \Gamma,\quad \text{with probability}\quad 1-p_{fail}
    \end{align}
using 
\begin{center}
$\CO\L( 2^{2^{\CO(\beta^4)} \mathrm{poly}(1/\epsilon\beta)}  \log (n/p_\mathrm{fail}) \R)$ copies of $\vrho$, and\\
$\CO\L(n\cdot 2^{2^{\CO(\beta^4)} \mathrm{poly}(1/\epsilon\beta)} \log(n/{p_\mathrm{fail}})\R)$ runtime. 
\end{center}
Furthermore, it performs coherent quantum measurements on at most $2^{\CO(\beta^4 + \max(\beta,1/d)\log(1/\epsilon))}$ qubits.  
\end{thm}
The dependence on the precision $1/\epsilon$ is exponential, arising from the volume of radius $\sim\log(1/\epsilon)$ on an expander graph. Still, for any constant $\epsilon,$ we only search for a constant-sized neighborhood, and each search is run completely independently of the others.
\begin{proof} 
We consider each of precision, sample complexity and runtime separately.

\vspace{0.1in}

\noindent\textbf{Precision guarantee.} For each site $i\in \Lambda,$ according to~\autoref{lem:exsitence}, there exist a guess Hamiltonians $\vK'_{\ell}$ such that, for all Pauli $\vA\in\{\vX_i,\vY_i,\vZ_i\}$, term $\vO= [\vA,\vP_{\gamma}]$ for $\gamma \sim \vA$, and Hamiltonians $\vG'_\ell$,
\begin{align}
\labs{Q(\vO,\vG'_\ell,\vA,\vK'_\ell)} \lesssim \frac{e^{\beta \Omega'/2}}{\sqrt{\beta}} \L(d e^{-\pi\ell/e^2d\beta} + \frac{\kappa \beta}{d}(V(\ell,\vA)+V(\ell,\vO)) \R)
\end{align}
With the parameter setting and sufficiently large $c_1 \leq c_2 \leq c_3$ in~\autoref{cond:param-simple},
\begin{align}
\labs{Q(\vO,\vG'_\ell,\vA,\vK'_\ell)} \leq\epsilon^2 \sqrt{\beta}/10\alpha.
\end{align}
Hence, the $\vK'_\ell$ returned in step 2 of the algorithm satisfies
\begin{align}
\labs{Q([\vA,\vP_\gamma],\vG'_\ell,\vA,\vK'_\ell)}  \leq \epsilon^2\sqrt{\beta}/5\alpha \qquad \text{for every } \vA,\gamma, \vG'_\ell. 
\end{align}
\autoref{lem:unique} then implies that
\begin{align}
    \labs{\braket{ \vO,[\vA,\vH-\vK'_\ell] }_{\vrho}} \lesssim \frac{e^{\beta \Omega'/2}}{\beta}  e^{-\pi\ell/e^2d\beta}+ \beta e^{-\Omega'/4d} d^{4+16e^2d^4\beta^2}  + \frac{\kappa}{d} e^{\beta \Omega'/2} ( V(\ell,\vO)+V(\ell,\vA)) + \epsilon^2/5\alpha.
\end{align}
With the parameter setting and sufficiently large $c_1 \leq c_2 \leq c_3$ in~\autoref{cond:param-simple} (more specifically, we first choose $c_1$, then $c_2$, then $c_3$),
\begin{align}
    \labs{\braket{ \vO,[\vA,\vH-\vK'_\ell] }_{\vrho}} \leq \epsilon^2/5\alpha + \epsilon^2/5\alpha+ \epsilon^2/5\alpha+ \epsilon^2/5\alpha \leq \epsilon^2/\alpha.
\end{align}
    Moreover,
    \begin{align}
        \e^{-200(d+q)\beta \log d\beta }\norm{[\vA,\vH-\vH']}_{\vtau}^2 \leq  \norm{[\vA,\vH-\vK'_\ell]}^2_{\vrho} \leq 2d\labs{\braket{ [\vA,\vH- (\vH_0 + \eta \vU'_{\ell_0})],\vO }_{\vrho}},
    \end{align}
    where the first inequality uses~\autoref{lem:KMSvariance} and the assumption $\beta d \geq 1$.
    Therefore, by~\autoref{lem:localcloseness} it holds for each $i$ and $\gamma \sim i$ that
    \begin{align}
        |h_\gamma - h_\gamma'| \leq  \e^{100(d+q)\beta \log d\beta } \sqrt{2d\epsilon^2/\alpha} \leq \epsilon.
    \end{align}
\vspace{0.1in}

\noindent\textbf{Sample complexity.} We want to measure $Q(\vO,\vG'_\ell,\vA,\vK'_\ell)$ for each single-site Pauli $\vA \in \{\vX_i,\vY_i,\vZ_i\}, \forall i$, each $\vO = [\vA,\vP_\gamma]$, and correspondingly each $\vG'_\ell$, $\vK'_\ell$ from the net $N_\kappa$. There are $3nd$ choices of the pair $\vA$ and $\vO$, and for each such choice there are $\lceil2/\kappa\rceil^{V(\ell)}$ choices of $\vK'_\ell$ and $\lceil2/\kappa\rceil^{V(\ell)d}$ choices of $\vG'_\ell$, totaling at most
\begin{align}
    3nd \lceil2/\kappa\rceil^{V(\ell)(d+1)} =\CO(n\cdot 2^{2^{\CO(\beta^4)} \mathrm{poly}(1/\epsilon)} )
\end{align}
operators $Q(\vO,\vG'_\ell,\vA,\vK'_\ell)$ to be measured. We can measure a large number of them simultaneously because each $Q(\vO,\vG'_\ell,\vA,\vK'_\ell)$ overlaps with at most
\begin{align}
    \chi \leq (d+1)V(2\ell)\lceil2/\kappa\rceil^{V(\ell)(d+1)} = 2^{2^{\CO(\beta^4)} \mathrm{poly}(1/\epsilon)}
\end{align}
others. Hence the sample complexity of the measurements (\autoref{alg:measureQ}) is
    \begin{align}
        \CO\left( \chi\cdot \frac{e^{\beta \Omega'}}{\beta (\sqrt{\beta}\epsilon^2/20\alpha)^2} \log (\CO(n\cdot 2^{2^{\CO(\beta^4)} \mathrm{poly}(1/\epsilon)} )/p_\mathrm{fail}) \right) = \CO\L( 2^{2^{\CO(\beta^4)} \mathrm{poly}(1/\epsilon)}  \log (n/p_\mathrm{fail}) \R)
    \end{align}

\vspace{0.1in}

\noindent\textbf{Runtime.} Direct substitution of the parameters from~\autoref{cond:param-simple} gives a runtime of $\CO(n\cdot 2^{2^{\CO(\beta^4)} \mathrm{poly}(1/\epsilon)} ) \cdot \mathrm{log}(n/p_\mathrm{fail})$.

\vspace{0.1in}

\noindent\textbf{Bootstrapping to the case $\beta < 1/d$.} As mentioned earlier, in this case we can rescale $\beta \leftarrow 1/d$ and $h_\gamma \leftarrow h_\gamma \cdot \beta d$. We apply the same algorithm as above, with precision redefined as $\epsilon \leftarrow \epsilon \cdot \beta d$.
\end{proof}

\subsection{An efficient high-precision learning algorithm for \texorpdfstring{$D$}{d}-dimensional lattices}
\label{sec:high_precision_algorithm}
We have seen that, to learn each local coefficient $h_{\gamma}$ to constant error $\epsilon=0.1$, it suffices to search over terms in a constant radius. However, at higher precisions $\epsilon \ll 1$, the algorithmic costs deteriorate super-polynomially due to the decay rate of Lieb-Robinson bounds on highly connected graphs. In this section, we show how to significantly improve the error dependence in the case of $D$-dimensional lattices (see \autoref{sec:Ham}) by a more refined locality estimate of Lieb-Robinson bounds. 

Suppose that we have already achieved a decent constant precision for every local coefficient (say, to an error of $0.1$). That is, we know that the ground truth $\vH$ satisfies
\begin{align}
\vH=\vH_0+\eta \vV,    
\end{align}
where $\vH_0$ is the current guess and $\vV$ is an unknown Hamiltonian (with the same interaction graph) such that each term $\|\vV_\gamma\| \leq 1$. Now, we would like to learn more information about $\vV$ and improve the learning error to $\eta/2$. Then, in our new guess $\vH_0+\eta \vU,$ not only do we have a smaller parameter space to search for $\vU$, but we also expect the identifiability observable $Q$ to depend most sensitively on closer terms. Based on this intuition, we propose a learning procedure that iteratively doubles the precision (see~\autoref{fig:local_learning}). Crucially, in each learning iteration, the search radius $\ell_0$ can be chosen to be \textit{independent} of the target learning error $\epsilon$ and only dependent on the geometry. Consider $Q(\vO,\vG,\vA,\vK)$ and $Q(\vO, \vG', \vA, \vK')$ with the inputs
\begin{align}
 \vG&=\vH_0+\eta \vW, \quad\vK=\vH_0+\eta \vU   \\
 \vG'&=\vH_0+\eta\vW_{\ell_0}, \quad \vK'=\vH_0+\eta \vU_{\ell_0}.
\end{align}
Then, by choosing a suitable search radius $\ell_0$ that only depends on $\beta$ and the geometry of the interaction, the truncation error will only contribute by a small fraction of the targeted error $\eta/2$ (by Item (B) of \autoref{lem:stability_Q}):
\begin{align}
    |Q(\vO,\vG,\vA,\vK) - Q(\vO, \vG', \vA, \vK')|
    &\lesssim \eta \frac{e^{\beta \Omega'/2}}{\sqrt{\beta}} \sum_{\ell=\ell_0}^{\infty} (S(\ell,\vA)+S(\ell,\vO)) (\beta+\frac{\ell}{d})(e^{-\ell^2/16e^4d^2\beta^2}+e^{-\pi\ell/2e^2d\beta})\\
    &= \text{(small factor independent of $\eta$)}\cdot\eta/2.
\end{align}

\begin{rmk} On an expander graph, the surface area scales as $\sim  d^{\ell+1}$, which grows faster than $e^{-\pi\ell/2e^2\beta}$ at low temperatures, and the RHS above is vacuous as an upper bound, hence the restriction to lattices. It is an interesting question to obtain a near-optimal learning algorithm for general graphs using a similar iterative approach.
\end{rmk}

Therefore, we have effectively reduced the problem to searching for $\vU_{\ell_0}$ over the radius $\ell_0$ such that $Q(\vO, \vG', \vA, \vK')$ is a fraction of $\eta$ for all $\vG'$ of the form $\vG'= \vH_0+ \eta\vW_{\ell_0}$. More precisely, let 
\begin{align}
\vU_{\ell_0} =\sum_{\gamma: \operatorname{dist}(\gamma, \vA) < \ell_0-1} u_\gamma \vP_\gamma\quad \text{and}\quad \vW_{\ell_0} =\sum_{\gamma: \operatorname{dist}(\gamma, \vO) < \ell_0-1} w_\gamma \vP_\gamma.\label{eq:UW}
\end{align}
The parameters $u_\gamma, w_\gamma$ are searched over a discrete epsilon net $N_{\kappa_0}$ of constant precision $\kappa_0$.

For bookkeeping, we display the choice of parameters of this section as follows, which are all independent of the system size $n$ and error $\epsilon$. In the present~\autoref{sec:high_precision_algorithm}, $\CO(\cdot)$ suppresses the dependence on the geometry (the lattice dimension $D$, degree bound $d$, and locality $q$). We introduce the suitable constants $c_1,c_2,c_3$ which depend only on $D,q,d$ so that the error analysis is strictly controlled by  `$\leq$' in later calculations (instead of $\lesssim$).

\begin{cond}\label{cond:param-iterative} Assume $\beta \geq 1/d$. In the rest of~\autoref{sec:high_precision_algorithm}, we set the following parameters. Let $\alpha =2d\e^{ 200(d+q)\beta \log d\beta} $ be an auxiliary parameter. For tunable constants $c_1 \ll c_2 \ll c_3$ (that may depend on $D,d,q$), 
    \begin{itemize}
    \item Frequency truncation: $\e^{\Omega'/4d} = c_1 \cdot \beta  d^{4+16\e^2d^4\beta^2}\alpha$ $ \Longrightarrow \Omega' =\CO(\beta^2) $.
    \item Search truncation radius: $\ell_0=c_2 \cdot 100D!d\beta (\beta \Omega' + \log (\alpha/\beta))  $ $ \Longrightarrow \ell_0 =\CO(\beta^4)$.
    \item Epsilon net precision: $\kappa_0=\frac{1}{c_3 \cdot \alpha} \ell_0^{-D-2} \e^{-\beta \Omega'/2} $ $\Longrightarrow  \kappa_0= \e^{-\CO(\beta^3)}$.
\end{itemize}

\end{cond}

\subsubsection{Existence and uniqueness of test Hamiltonians under perturbation}
Here, we derive the analog of the existence and uniqueness property of $Q$ (\autoref{sec:exist_unique}), assuming that the guess is already pretty good. We only require the bounds to be a fraction of $\eta$, and a much smaller search radius suffices.

\begin{lem}[Existence of good local guess $\vU_{\ell_0}$ on the epsilon net] 
\label{lem:existence-iterative} Assume that $\beta \geq 1/d$. Consider the parameter choice from~\autoref{cond:param-iterative}. 
For every $\vA,\vO$ such that $\norm{\vA},\norm{\vO}\le 1$, there exists a $\vU'_{\ell_0} =\sum_{\gamma: \operatorname{dist}(\gamma, \vA) < \ell } u'_\gamma \vP_\gamma$ with $u'_{\gamma} \in N_{\kappa_0}$ such that for every $\vW'_{\ell_0} =\sum_{\gamma: \operatorname{dist}(\gamma, \vO) < \ell } w'_\gamma \vP_\gamma$ with $w'_{\gamma} \in N_{\kappa_0}$, the identifiability observable satisfies
\begin{align}\labs{Q(\vO,\vH_0 + \eta \vW'_{\ell_0},\vA,\vH_0 + \eta \vU'_{\ell_0})} \leq \eta\frac{\sqrt{\beta}}{20\alpha}.
\end{align}
\end{lem}
\begin{proof}
    The proof is similar to that of~\autoref{lem:exsitence}. First, there exists a global $\vV$ such that $\vH_0 + \eta \vV=\vH$, and therefore $Q(\vO, \vG, \vA, \vH_0 + \eta \vV)=0$ for each $\vO,\vG,\vA$ due to~\autoref{lem:existence-global}. Next, Item (B) of~\autoref{lem:stability_Q} with the truncation radius $\ell_0$ guarantees that
    \begin{align}
        &\L|Q(\vO,\vG,\vA,\vH_0 + \eta \vV) - Q(\vO, \vG, \vA, \vH_0 + \eta \vV_{\ell_0})\R| \\
        &\lesssim \eta \frac{e^{\beta \Omega'/2}}{\sqrt{\beta}} \sum_{\ell=\ell_0}^{\infty} (S(\ell,\vA)+S(\ell,\vO)) (\beta+\frac{\ell}{d})(e^{-\ell^2/16e^4d^2\beta^2}+e^{-\pi\ell/2e^2d\beta})\\
    &\lesssim \eta \frac{e^{\beta \Omega'/2}}{\sqrt{\beta}} \sum_{\ell=\ell_0}^{\infty} \CO(\ell^{D-1}) (\beta+\frac{\ell}{d})(e^{-\ell^2/16e^4d^2\beta^2}+e^{-\pi\ell/2e^2d\beta})\\
    &\le  \CO(1)\cdot \eta \frac{e^{\beta \Omega'/2}}{\sqrt{\beta}} \L(\frac{\pi \ell_0 D!}{2e^2}\R)^{D} e^{-\pi\ell_0/2e^2d\beta} \tag*{(assuming $\ell_0 \geq 100 D!d\beta$)}\\
    & \le \CO(1)\cdot \eta\frac{\sqrt{\beta}}{
        40\alpha},
    \end{align} 
    for each $\vG$, including $\vH_0 + \eta \vW'_{\ell_0}$ from the net $N_{\kappa_0}$.
    Finally, we use Item (C) of~\autoref{lem:stability_Q} to round the above $\vV_{\ell_0}$ to the epsilon net to obtain $\vU'_{\ell_0}$ such that
    \begin{align}
        \L|Q(\vO,\vG,\vA,\vH_0 + \eta \vU_{\ell_0}) - Q(\vO, \vG, \vA, \vH_0 + \eta \vU'_{\ell_0})\R| &\lesssim \frac{\eta \kappa_0\sqrt{\beta}}{d} e^{\beta \Omega'/2} ( V(\ell_0,\vO)+V(\ell_0,\vA) ) \\
        &\le \CO(1) \frac{\eta \kappa_0\sqrt{\beta}}{d} e^{\beta \Omega'/2}  \ell_0^{D+1}
        \le  \CO(1) \cdot \eta\sqrt{\beta}/
        40\alpha.
    \end{align}
    For any fixed choice of $c_1$ (fixed $\Omega'$), we can choose the constants $ c_2 \ll c_3$ from~\autoref{cond:param-iterative} to be sufficiently large (we first choose $c_2 \gg c_1$, and then $c_3 \gg c_2$) such that $\lesssim$ can be replaced by proper $\leq$ in the above bounds.
    Finally, we collect the error terms to conclude the proof.
\end{proof}

\begin{lem}[KMS-local identifiability]\label{lem:KMS-local-identify} Assume that $\beta \geq 1/d$.
In the same setting as~\autoref{lem:existence-iterative}, suppose there is a local guess $\vU'_{\ell_0}$ from the epsilon net $N_{\kappa_0}$ such that $\labs{Q(\vO,\vH_0 + \eta \vW'_{\ell_0},\vA,\vH_0 + \eta \vU'_{\ell_0})} \leq \eta\sqrt{\beta}/10\alpha $ for every $\vW'_{\ell_0}$. Then, 
    \begin{align}
    \labs{\braket{ \vO,[\vA,\vH- (\vH_0 + \eta \vU'_{\ell_0})] }_{\vrho}} \leq  \frac{\eta}{5\alpha}.
\end{align}
\end{lem}
\begin{proof} Recall~\autoref{lem:truncating_bohr} and that we can write $\vH = \vH_0 + \eta \vV$,
\begin{align}
    \frac{\beta\sqrt{2\sigma\sqrt{2\pi}}}{2}\braket{ \vO, [\vA,\vH-(\vH_0 + \eta \vU'_{\ell_0})] }_{\vrho}= Q(\vO,\vH,\vA,\vH_0 + \eta \vU'_{\ell_0})+\frac{\eta \beta }{2}\int_{\labs{\omega'}\ge \Omega'} \braket{\vO,[\hat{\vA}_{\vH'}(\omega'),\vV-\vU'_{\ell_0}]}_{\vrho} \rd \omega'.
\end{align}
The high-frequency part can be bounded using~\autoref{lem:high_freq}, in particular with the choice of $\Omega'$ from~\autoref{cond:param-iterative} it is bounded by 
\begin{align}
    \left|\frac{\eta \beta }{2}\int_{\labs{\omega'}\ge \Omega'} \braket{\vO,[\hat{\vA}_{\vH'}(\omega'),\vV-\vU'_{\ell_0}]}_{\vrho} \rd \omega' \right|\le \CO(1)\cdot \eta \beta^{3/2} e^{-\Omega'/4d} d^{4+16e^2d^4\beta^2}  \le \CO(1)\cdot \eta\sqrt{\beta}/30\alpha.
\end{align}
We can choose the constant $c_1$ from~\autoref{cond:param-iterative} to be sufficiently large to obtain a proper $\leq$ bound
\begin{align}
    \left|\frac{\eta \beta }{2}\int_{\labs{\omega'}\ge \Omega'} \braket{\vO,[\hat{\vA}_{\vH'}(\omega'),\vV-\vU'_{\ell_0}]}_{\vrho} \rd \omega' \right| \leq \eta\sqrt{\beta}/30\alpha.
\end{align}
Next, using the stability of $Q$ (Item (B) of~\autoref{lem:stability_Q}) with the truncation radius $\ell_0$ from~\autoref{cond:param-iterative}, and for brevity letting $\vH'= \vH_0 + \eta \vU'_{\ell_0}$, we have
\begin{align}
    \labs{Q(\vO,\vH_0 + \eta \vV,\vA,\vH'} &\le \labs{Q(\vO, \vH_0 + \eta \vV,\vA,\vH')-Q(\vO, \vH_0 + \eta \vV_{\ell_0},\vA,\vH')}+ \labs{Q(\vO,\vH_0 + \eta \vV_{\ell_0},\vA,\vH')}\\
    &\le \CO(1) \cdot \eta\sqrt{\beta}/30\alpha + \labs{Q(\vO,\vH_0 + \eta \vV_{\ell_0},\vA,\vH')},
\end{align}
where the second line uses Item (B) of~\autoref{lem:stability_Q}
\begin{align}
    &\labs{Q(\vO, \vH_0 + \eta \vV,\vA,\vH')-Q(\vO, \vH_0 + \eta \vV_{\ell_0},\vA,\vH')}  \\ 
    &\lesssim \eta \frac{e^{\beta \Omega'/2}}{\sqrt{\beta}} \sum_{\ell=\ell_0}^{\infty} (S(\ell,\vA)+S(\ell,\vO)) (\beta+\frac{\ell}{d})(e^{-\ell^2/16e^4d^2\beta^2}+e^{-\pi\ell/2e^2d\beta})\\
    &\le \CO(1) \cdot \eta \frac{e^{\beta \Omega'/2}}{\sqrt{\beta}} \sum_{\ell=\ell_0}^{\infty} \ell^{D-1} (\beta+\frac{\ell}{d})(e^{-\ell^2/16e^4d^2\beta^2}+e^{-\pi\ell/2e^2d\beta})\\
     &\le  \CO(1)\cdot \eta \frac{e^{\beta \Omega'/2}}{\sqrt{\beta}} \L(\frac{\pi \ell_0 D!}{2e^2}\R)^{D} e^{-\pi\ell_0/2e^2d\beta} \tag*{(assuming $\ell_0 \geq 100 D!d\beta$)}\\
    &\le \CO(1) \cdot \eta\sqrt{\beta}/30\alpha .
\end{align}
For a fixed $\Omega'$ we can choose $c_2 \gg c_1$ from~\autoref{cond:param-iterative} to be sufficiently large to obtain a proper $\leq$ bound
\begin{align}
    \labs{Q(\vO,\vH_0 + \eta \vV,\vA,\vH'}
    \leq \eta\sqrt{\beta}/30\alpha + \labs{Q(\vO,\vH_0 + \eta \vV_{\ell_0},\vA,\vH')}.
\end{align}

Finally we use Item (C) of~\autoref{lem:stability_Q} to round $\vV_{\ell_0}$ to the epsilon net. There exists a $\vW'_{\ell_0}$ from the net $N_{\kappa_0}$ such that
\begin{align}
    \labs{Q(\vO,\vH_0 + \eta \vV_{\ell_0},\vA,\vH') - Q(\vO,\vH_0 + \eta \vW'_{\ell_0},\vA,\vH')} &\lesssim \frac{\eta\kappa_0 \sqrt{\beta}}{d} e^{\beta \Omega'/2} ( V(\ell_0,\vO)+V(\ell_0,\vA))\\
    & \le \CO(1) \frac{\eta\kappa_0 \sqrt{\beta}}{d} e^{\beta \Omega'/2} \ell_0^{D+1}  \le \CO(1)\cdot\eta \sqrt{\beta}/30\alpha.
\end{align}
For fixed $\Omega', \ell_0$, we can choose $c_3 \gg c_2$ from~\autoref{cond:param-iterative} to be sufficiently large to obtain
\begin{align}
    \labs{Q(\vO,\vH_0 + \eta \vV_{\ell_0},\vA,\vH') - Q(\vO,\vH_0 + \eta \vW'_{\ell_0},\vA,\vH')} \leq \eta \sqrt{\beta}/30\alpha.
\end{align}

Finally, we collect the error terms
\begin{align}
    \frac{\sqrt{2\beta\sqrt{2\pi}}}{2}|\braket{ \vO, [\vA,\vH-(\vH_0 + \eta \vU'_{\ell_0})] }_{\vrho}| \leq \eta\sqrt{\beta}/10\alpha 
\end{align}
and rearrange to conclude the proof.
\end{proof}

\subsubsection{The algorithm}

We now describe one iteration step that will reduce the current error $\eta$ to $\eta/2,$ using quasi-local measurements and local search. The main difference from the non-iterative algorithm (\autoref{alg:constant_err}) is that we are given a good guess $\vH_0$ already, and we only aim for doubling the precision in one iteration. While the measurements still involve $\log^D(1/\epsilon)$-sized neighborhood, we only vary the Hamiltonian over a much smaller radius $\ell_0$, on top of a background $\vH_0$. 

\begin{alg}[One iteration step for $D$-dimensional lattices]
\label{alg:learning-iteration} Consider the Gibbs state $\vrho$ for a $D$-dimension Hmailtonian $\vH$ at inverse temperature $\beta$ (\autoref{sec:Ham}). Suppose ground truth satisfies $\vH=\vH_0+\eta \vV$ for a known $\vH_0 = \sum_{\gamma\in \Gamma} h_{0,\gamma}\vP_{\gamma}$ and an unknown $\vV$ with same interaction graph as $\vH$ such that each term $\|\vV_\gamma\| \leq 1$. 

\begin{enumerate}
\item (Define identifiability observables) Set $\Omega', \ell_0, \kappa_0$ according to~\autoref{cond:param-iterative}. For each site $i\in \Lambda,$ and $\vA\in\{\vX_i,\vY_i,\vZ_i\}$ and adjacent terms $\gamma: \gamma \sim i$: 
\begin{itemize}
    \item Consider Hamiltonians of the form $\vG'= \vH_0 + \eta \vW'_{\ell_0}$ and $\vK'= \vH_0 + \eta \vU'_{\ell_0}$ where $\vW'_{\ell_0}$ (and $\vU'_{\ell_0}$) are the Hamiltonians supported on $V(\ell_0+1,\vA)$ (and $V(\ell_0,\vA)$, respectively), whose coefficients are taken from the net $N_{\kappa_0}$.
    \item Take a truncation radius $\ell =\CO(\beta^2 \Omega' +\beta \log\frac{\alpha}{\sqrt{\beta} \eta}) =\CO(\beta^4 + \beta \log 1/\eta)$ and consider $\vG'\rightarrow \vG'_\ell$, $\vK'\rightarrow \vK'_\ell$.
    \end{itemize}

    \item (Measure in parallel) Perform experiments (\autoref{alg:measureQ}) to measure all observables defined above to error $\epsilon=\sqrt{\beta} \eta/80\alpha$ and failure probability $p_\mathrm{fail}$
    \begin{align}
        Q_{exp}([\vA,\vP_{\gamma}],\vG'_{\ell},\vA,\vK'_{\ell}) \quad \text{over inputs}\quad i\in \Lambda, \quad\vA\in\{\vX_i,\vY_i,\vZ_i\}, \quad \gamma: \gamma \sim i, \quad \vW'_{\ell_0}, \vU'_{\ell_0} \in N_{\kappa_0}.
    \end{align}
    \item (Identify local terms) For each site $i$: 
    \begin{itemize}
        \item  Return the parameters of the Hamiltonian $\vU'_{\ell_0}$ which attains the weakest $Q$ for all $\vA,\gamma,\vG'_{\ell}$
    \begin{align}
        \min_{\vU'_{\ell_0}} \max_{\vA,\gamma,\vW'_{\ell_0}} \labs{Q_{exp}([\vA,\vP_{\gamma}],\vG'_{\ell},\vA,\vK'_{\ell})}.
    \end{align}
    \item     Record the Hamiltonian terms in $\vU'_{\ell_0}$ that acts on qubit $i$, and set 
    \begin{align}
        h'_{\gamma} \leftarrow h_{0,\gamma} + \eta  u_\gamma.
    \end{align}
    \end{itemize}
    \item Return the collection of coefficients $\{h'_{\gamma}\}_{\gamma \in \Gamma}.$
\end{enumerate}

\end{alg}

The precision-independent search radius yields significant reductions in the search space and improves the algorithmic costs.
\begin{thm}[Cost per learning iteration]\label{thm:learning-iteration} Assume $\beta \geq 1/d$.
    With probability $1-p_\mathrm{fail}$, \autoref{alg:learning-iteration} outputs the coefficients of the Hamiltonian $\vH'$ with error $\eta/2$ from those of $\vH$, using 
    \begin{center} 
    $\CO\L( \frac{e^{\CO(\beta^{cD})}}{\eta^2} \log^{D}(1/\eta) \log (n/p_\mathrm{fail})\R)$ copies of $\vrho$, and\\
     $\CO\L(n\log (n/p_\mathrm{fail})\cdot \frac{e^{\CO(\beta^{c'D})}}{\eta^2} \log^{c''D}(1/\eta)\R)$ runtime 
    \end{center}
    (including both quantum gate count and classical processing). Furthermore, it only involves coherent quantum measurements on at most $\CO((\beta^4+\beta \log 1/\epsilon)^D)$ qubits. Here, $c,c',c''$ are absolute constants.
\end{thm}
\begin{proof}

We consider each of precision, sample complexity and runtime separately.

\vspace{0.1in}

\noindent\textbf{Precision guarantee.} Consider site $i\in \Lambda,$ Pauli $\vA\in\{\vX_i,\vY_i,\vZ_i\}$, and term $\vO= [\vA,\vP_{\gamma}]$ for $\gamma \sim i$.   
    For guess Hamiltonians $\vG' = \vH_0 + \eta \vW'_{\ell_0}$, $\vK' =\vH_0 + \eta \vU'_{\ell_0}$ and the truncated versions $\vG'_{\ell}$ and $\vK'_{\ell}$, according to Item (A) of~\autoref{lem:stability_Q},   \begin{align}
        \labs{Q(\vO,\vG',\vA,\vK') - Q(\vO, \vG'_{\ell},\vA,\vK'_{\ell})}&\lesssim  \frac{e^{\beta \Omega'/2}}{\sqrt{\beta}} (e^{-\ell^2/16e^4d^2\beta^2}+e^{-\pi\ell/e^2d\beta}) \lesssim \sqrt{\beta} \eta/80\alpha.
    \end{align}
    Choosing the constants in $\ell=\CO(\beta^4 + \beta \log 1/\eta)$ to be sufficiently large we obtain a proper bound
    \begin{align}
    \labs{Q(\vO,\vG',\vA,\vK') - Q(\vO, \vG'_{\ell},\vA,\vK'_{\ell})} \leq \sqrt{\beta} \eta/80\alpha.\label{eq:1221}
    \end{align}
    
    In addition,~\autoref{alg:measureQ} guarantees that with probability at least $1-p_\mathrm{fail}$, the experiment estimate satisfies
    \begin{align}
        |Q_{exp}([\vA,\vP_{\gamma}],\vG'_{\ell},\vA,\vK'_{\ell})-Q(\vO, \vG'_{\ell},\vA,\vK'_{\ell})| \leq \sqrt{\beta} \eta/80\alpha.
        \label{eq:1222}
    \end{align}
     Combining~\eqref{eq:1221},~\eqref{eq:1222} with~\autoref{lem:existence-iterative}, it follows that, with probability $1-p_\mathrm{fail}$, the $\vU'_{\ell_0}$ returned in step 1 of~\autoref{alg:learning-iteration} satisfies
    \begin{align}
        \labs{Q([\vA,\vP_{\gamma}],\vH_0 + \eta \vW'_{\ell_0},\vA,\vH_0 + \eta \vU'_{\ell_0})} \leq \eta\sqrt{\beta}/10\alpha  \qquad \text{for every } \vA, \gamma, \vW'_{\ell_0}.
    \end{align}
    Then,~\autoref{lem:KMS-local-identify} implies that
    \begin{align}
         \max_{\vO=[\vA,\vP_{\gamma}] }\labs{\braket{ [\vA,\vH- (\vH_0 + \eta \vU'_{\ell_0})],\vO }_{\vrho}} \leq \eta/5 \alpha.
         \label{eq:1223}
    \end{align}
Additionally recall that
\begin{align}
        \norm{[\vA,\vH-(\vH_0 + \eta \vU'_{\ell_0})]}^2_{\vrho} \leq 2d\eta \max_{\vO=[\vA,\vP_{\gamma}] } \labs{\braket{ [\vA,\vH- (\vH_0 + \eta \vU'_{\ell_0})],\vO }_{\vrho}}.
        \label{eq:1224}
    \end{align}
Now, we convert the KMS norm to the error in local coefficients. Combining~\eqref{eq:1223},~\eqref{eq:1224}, we get, for each term $\gamma$ acting on qubit $i$ and $\vA \in \{\vX_i,\vY_i, \vZ_i\}$, that
    \begin{align}
          \norm{[\vA,\vH-\vH']}_{\vtau} &\le \e^{100(d+q) \beta \log d\beta } \norm{[\vA,\vH-\vH']}_{\vrho}. \tag*{(\autoref{lem:KMSvariance} and $\beta d \geq 1$)}\\
         &\leq \e^{100(d+q)\beta \log d\beta } \eta \sqrt{2d/5\alpha}.
    \end{align}
    It follows from~\autoref{lem:localcloseness} that
    \begin{align}
    |h_{\gamma} -(h_{0,\gamma} + \eta  u_\gamma)| \leq \eta/2.  
    \end{align}

\vspace{0.1in}    

\noindent\textbf{Sample complexity.} We want to measure $Q(\vO,\vG'_\ell,\vA,\vK'_\ell)$ for each single-site Pauli $\vA \in \{\vX_i,\vY_i,\vZ_i\}, \forall i$, each $\vO = [\vA,\vP_\gamma]$, and correspondingly each $\vW'_{\ell_0}$, $\vU'_{\ell_0}$ from the net $N_{\kappa_0}$. There are $3nd$ choices of the pair $\vA$ and $\vO$, and for each such choice there are $(2/\kappa_0)^{V(\ell_0)}$ choices of $\vU'_{\ell_0}$ and $(2/\kappa_0)^{dV(\ell_0)}$ choices of $\vW'_{\ell_0}$, totaling at most $|S|=3nd (2/\kappa_0)^{V(\ell_0)(d+1)}=n\cdot e^{\CO(\beta^{cD})}$ operators $\vQ(\vO,\vG'_\ell,\vA,\vK'_\ell)$ to be measured. Furthermore, each $\vQ(\vO,\vG'_\ell,\vA,\vK'_\ell)$ overlaps with at most $\chi = (d+1)V(2\ell)(2/\kappa_0)^{V(\ell_0)(d+1)} =e^{\CO(\beta^{cD})} \log^D (1/\eta)$ others. Here $c$ is an absolute constant.
Hence,~\autoref{alg:measureQ} needs a sample complexity of
    \begin{align}
       \CO(\chi \frac{e^{\beta \Omega'/2}}{\sqrt{\beta} (\sqrt{\beta} \eta/80\alpha)^2} \log (n\cdot e^{\CO(\beta^{cD})}/p_\mathrm{fail}))= \frac{e^{\CO(\beta^{cD})}}{\eta^2} \log^{D}(1/\eta) \log (n/p_\mathrm{fail}).
    \end{align}

\noindent\textbf{Runtime.} Direct substitution of the parameters from~\autoref{cond:param-iterative} and the measurement truncation radius $\ell=\CO(\beta^4 + \beta \log 1/\eta)$ gives a time complexity of 
    $$|S|\mathrm{poly}(\beta, V(\ell),\log(1/\kappa_0 \eta), \frac{e^{\beta \Omega'}\|\vA\|\|\vO\|}{\sqrt{\beta}}) \frac{\log (|S|/p_\mathrm{fail})}{\eta^2}=\CO(n\cdot \frac{e^{\CO(\beta^{c'D})}}{\eta^2} \log^{c''D}(1/\eta)\log (n/p_\mathrm{fail})),$$
    where $c',c''$ are absolute constants.
\end{proof}

Finally, we can chain the iteration step to obtain the full algorithmic cost. For an error $\epsilon,$ the number of iterations scales only logarithmically $\log(1/\epsilon),$ and the $\epsilon$ dependence is dominated by the measurement costs $1/\epsilon^2.$ 
\begin{thm}[Learning lattice Hamiltonians near-optimally in $\epsilon$ and $n$ - Thm \ref{thm:latticethm}] Chaining~\autoref{alg:learning-iteration}, we can learn the Hamiltonian for quantum Gibbs states on $D$-dimensional lattices to precision $\epsilon$, with probability $1-p_\mathrm{fail}$ using 
\begin{center}
$\CO\L(\frac{e^{\CO(\beta^{cD})}}{\beta^2 \varepsilon^2} (\log 1/\varepsilon)^{D+1} \log (n/p_{\mathrm{fail}})\R)$ samples, and\\ 
$\CO\L(n\log (n/p_\mathrm{fail})\cdot \frac{e^{\CO(\beta^{c'D})}}{\beta^2 \varepsilon^2} \log^{c''D}( 1/\varepsilon) \R)$ runtime 
\end{center}
(including both quantum gate count and classical processing).
Furthermore, it performs coherent quantum measurements on at most $\CO((\beta^4 + \max(\beta,1/d)\log 1/\epsilon)^D)$ qubits. Here, $c,c',c''$ are absolute constants.
\end{thm}
\begin{proof}
    Iteratively apply~\autoref{alg:learning-iteration} with $p_\mathrm{fail}=1/O(\log(1/\varepsilon))$ whose performance guarantee and complexity are given in~\autoref{thm:learning-iteration}. Repeating $\CO(\log 1/\varepsilon)$ iterations suffices.
    \vspace{0.1in}
    
\noindent\textbf{Bootstrapping to the case $\beta < 1/d$.} We can rescale $\beta \leftarrow 1/d$ and $h_\gamma \leftarrow h_\gamma \cdot \beta d$. We apply the same algorithm as above, with precision redefined as $\epsilon \leftarrow \epsilon \cdot \beta d$.
\end{proof}

\acknowledgments
We thank Thiago Bergamaschi, Jonas Haferkamp, Yunchao Liu, Daniel Mark, Weiliang Wang, and Qi Ye for helpful discussions. We thank Cambyse Rouze for collaboration in the recent related work~\cite{chen2025Markov}. CFC is supported by a Simons-CIQC postdoctoral fellowship through NSF QLCI Grant No. 2016245. AA and QTN acknowledge support through the NSF Award No. 2238836. AA acknowledges support through the NSF award QCIS-FF: Quantum Computing \& Information Science Faculty Fellow at Harvard University (NSF 2013303), and NSF Award No. 2430375. QTN acknowledges support through the Harvard Quantum Initiative PhD fellowship and IBM PhD fellowship.

\bibliographystyle{alphaUrlePrint.bst}
\bibliography{ref}

\appendix
\section{Standard measurement costs}\label{sec:measurement_proof}
Here, we collect routine quantum algorithm arguments for performing measurements and time evolution.
\begin{proof}[Proof of~\autoref{lem:measure-single-Q}]
    It is a standard result that for an observable $\vE$ with $\|\vE\| \leq 1$ we can estimate $\tr(\vrho \vE)$ to within additive error $\epsilon$ with probability $1- \delta$ using $\CO(\log (1/\delta)/\epsilon^2)$ copies of $\vrho$.
    Here note that for each $\vO, \vA$, $\vG'_\ell$, $\vK'_\ell$, the operator
    \begin{align}
    \vQ(\vO,\vG'_\ell,\vA,\vK'_\ell) &= \frac{1}{\sqrt{2\pi}}\int_{-\infty}^{\infty}\int_{-\infty}^{\infty} \big(  \vO^{\dagger}_{\vG_\ell}(t) \vA_{\vK_\ell}(t'+t) h_+(t') - \vA_{\vK_\ell}(t'+t) \vO^{\dagger}_{\vG_\ell}(t) h_-(t') \big)  g_{\beta}(t)  \rd t'\rd t
    \end{align}
    is $V(\ell)(|\Supp{\vA}|+|\Supp{\vO}|)$-local and has bounded norm
    \begin{align}
    \|\vQ(\vO,\vG'_\ell,\vA,\vK'_\ell)\| &\leq \|\vA\| \|\vO\|\frac{1}{\sqrt{2\pi}}\int_{-\infty}^{\infty}\int_{-\infty}^{\infty} \big(|h_+(t')| + |h_-(t')| \big)  g_{\beta}(t)  \rd t'\rd t \\
    &\lesssim \|\vA\| \|\vO\|\int_{-\infty}^{\infty}\int_{-\infty}^{\infty}  e^{-\sigma^2t'^2} \frac{\sqrt{\sigma}}{\beta}e^{\beta \Omega'/2+\sigma^2\beta^2/4}\frac{4}{\beta} e^{-2\pi |t|/\beta} \rd t'\rd t\lesssim \frac{e^{\beta \Omega'/2}}{\sqrt{\beta}} \|\vA\| \|\vO\|.
    \end{align}
    So estimating $\tr(\vQ(\vO,\vG'_\ell,\vA,\vK'_\ell) \vrho)$ requires $\CO(\frac{e^{\beta \Omega'} \|\vA\|^2 \|\vO\|^2}{\beta \epsilon^2} \log (1/\delta))$ samples.

    Now we look more carefully at the gate complexity of measuring $Q(\vO,\vG'_\ell,\vA,\vK'_\ell)$. A trivial gate complexity upperbound for a measurement shot is $2^{V(\ell,\vO)+ V(\ell,\vA)}$, which is $2^{d^{\CO(\ell)}}$ for general spare graphs and $2^{\ell^{\CO(D)}}$ for $D$-dimensional lattices. Since we will be interested in optimal learning on $D$-dimensional lattices and in $\ell$ scaling logarithmically with the inverse learning precision, we would like to improve this trivial gate complexity to $\mathrm{poly}\ V(\ell)$. We will treat $d$ as constant.
    
    The first step is to truncate the time integrals. According to~\autoref{lem:time_truncation},
    \begin{align}
&\labs{\frac{1}{\sqrt{2\pi}}\L(\int_{-\infty}^{\infty}\int_{-\infty}^{\infty}-\int_{\labs{t}\le T}\int_{\labs{t'}\le T'}\R) \bigg( h_+(t') \vO^{\dagger}_{\vG'_\ell}(t) \vA_{\vK'_\ell}(t'+t)  - h_-(t')\vA_{\vK'_\ell}(t'+t) \vO^{\dagger}_{\vG'_\ell}(t) \bigg)  g_{\beta}(t)\rd t'\rd t} \\
&\lesssim \|\vO\|\|\vA\| \frac{e^{\beta\Omega'/2 + \sigma^2\beta^2/4}}{\sqrt{\sigma}\beta} \left(e^{-\sigma^2T'^2}+ e^{-2\pi T/\beta}\right).
\end{align}    
It suffices to choose $T=\CO( \beta^2 \Omega' + \beta \log(1/\sqrt{\beta}\epsilon)) $, $T'=O( \beta^3 \Omega' + \beta^2 \log(1/\sqrt{\beta}\epsilon))^{1/2}$.

Next, we discretize the truncated operator. Define
\begin{align}
    \tilde{\vQ}_{\Delta} &= \sum_{j=0}^{\lceil T'/\Delta \rceil} \sum_{k=0}^{\lceil T/\Delta \rceil} \big(  \vO^{\dagger}_{\vG'_\ell}(j\Delta) \vA_{\vK'_\ell}((j+k)\Delta) h_+(k\Delta) - \vA_{\vK'_\ell}((j+k)\Delta) \vO^{\dagger}_{\vG'_\ell}(j\Delta) h_-(k\Delta) \big)  g_{\beta}(j\Delta)  \Delta^2\\
    &= \sum_{j=0}^{\lceil T'/\Delta \rceil} \sum_{k=0}^{\lceil T/\Delta \rceil} \tilde{\vQ}_{j,k}.
\end{align} 
which satisfies $\|\tilde{\vQ}_{\Delta} -  \vQ(\vO,\vG'_\ell,\vA,\vK'_\ell) \| \leq \epsilon$ when $\Delta = 1/\mathrm{poly}(\beta, T,T',\frac{e^{\beta \Omega'/2}}{\sqrt{\beta}})$.
Now, $\tilde{\vQ}_{\Delta}$ can be implemented efficiently using standard Hamiltonian simulation and block-encoding tools~\cite{gilyen2019quantum}, so we will be brief. Block-encodings of $\vG'_\ell$, $\vK'_\ell$ can be implemented to precision $\epsilon$ using $\mathrm{poly}(V(\ell), \log( 1/\kappa\epsilon))$ elementary gates. Hence, block-encodings of time-$t$ Heisenberg evolutions of $\vO, \vA$ can be implemented using $\mathrm{poly}(tV(\ell), \log(1/\kappa\epsilon))$~\cite{gilyen2019quantum, low2017optimal}. Then taking a linear combination of these Heisenber-evolved operators with the coeffcients $\pm h_\pm(k\Delta) g_\beta(j\Delta)\Delta^2$ requires $\mathrm{poly}(\frac{e^{\beta \Omega'/2}}{\sqrt{\beta}}(T+T')/\Delta  )$ additional gates. Therefore, in total we require $\mathrm{poly}(\beta, V(\ell),\log(1/\kappa \epsilon), \frac{e^{\beta \Omega'/2}}{\sqrt{\beta}}\|\vA\|\|\vO\|)$ elementary gates to obtain one measuremet shot for $\vQ(\vO,\vG'_\ell,\vA,\vK'_\ell)$.

\end{proof}

\subsection{Time truncation of \texorpdfstring{$Q$}{q}}
\begin{lem}[Truncating the time integral]
\label{lem:time_truncation}
In the setting of~\autoref{lem:truncating_bohr},
\begin{align}
    &\labs{\frac{1}{\sqrt{2\pi}}\L(\int_{-\infty}^{\infty}\int_{-\infty}^{\infty}-\int_{\labs{t}\le T}\int_{\labs{t'}\le T'}\R) \bigg( h_+(t') \vO^{\dagger}_{\vH}(t) \vA_{\vH'}(t'+t)  - h_-(t')\vA_{\vH'}(t'+t) \vO^{\dagger}_{\vH}(t) \bigg)  g_{\beta}(t)\rd t'\rd t} \\
&\lesssim \|\vO\|\|\vA\| \frac{e^{\beta\Omega'/2 + \sigma^2\beta^2/4}}{\sqrt{\sigma}\beta} \left(e^{-\sigma^2T'^2}+ e^{-2\pi T/\beta}\right). 
\end{align}    
\begin{proof}
Consider,
\begin{align}
    &\labs{\frac{1}{\sqrt{2\pi}}\L(\int_{-\infty}^{\infty}\int_{-\infty}^{\infty}-\int_{\labs{t}\le T}\int_{\labs{t'}\le T'}\R) \bigg( h_+(t') \vO^{\dagger}_{\vH}(t) \vA_{\vH'}(t'+t)  - h_-(t')\vA_{\vH'}(t'+t) \vO^{\dagger}_{\vH}(t) \bigg) g_{\beta}(t)\rd t'\rd t} \\
&\le \norm{\vO}\norm{\vA}\frac{1}{\sqrt{2\pi}}\int_{-\infty}^{\infty}\int_{-\infty}^{\infty}\L( \indicator(\labs{t}\ge T) +\indicator(\labs{t'}\ge T')\R) \cdot \L( \labs{h_+(t')} + \labs{h_-(t')} \R) \labs{g_{\beta}(t)}\rd t'\rd t. 
\end{align} 
Recall that 
$$\labs{h_+(t')},\labs{h_-(t')} \lesssim e^{-\sigma^2t'^2} \frac{\sqrt{\sigma}}{\beta}e^{\beta \Omega'/2+\sigma^2\beta^2/4}$$
and 
$$|g_{\beta}(t)|=\frac{2\pi^{3/2}}{4\beta(1+\cosh(2\pi t/\beta))} \leq \frac{4}{\beta} e^{-2\pi |t|/\beta}.$$
By direct substitution, we get
$$\int_{-\infty}^{\infty}\int_{-\infty}^{\infty}\L( \indicator(\labs{t}\ge T))\R) \cdot \L( \labs{h_+(t')} + \labs{h_-(t')} \R) \labs{g_{\beta}(t)}\rd t'\rd t\leq \frac{2e^{\beta\Omega'/2 + \sigma^2\beta^2/4}}{\sqrt{\sigma}\beta \sqrt{\pi}}e^{-2\pi T/\beta},$$
and 
$$\int_{-\infty}^{\infty}\int_{-\infty}^{\infty}\L( \indicator(\labs{t'}\ge T'))\R) \cdot \L( \labs{h_+(t')} + \labs{h_-(t')} \R) \labs{g_{\beta}(t)}\rd t'\rd t\leq \frac{4}{2\sqrt{\sigma}\pi\beta}e^{-\sigma^2T'^2} e^{\beta \Omega'/2+\sigma^2\beta^2/4}.$$
\end{proof}
\end{lem}
\section{Lieb-Robinson estimates}\label{sec:LRbounds}

In this section, we recall some standard Lieb-Robinson estimates (see, e.g.,~\cite{chen2023speed}). The main subroutines are as follows.
\begin{lem}[Lieb-Robinson bound]\label{lem:LRbounds}
For a Hamiltonian $\vH = \sum_{\gamma}\vh_{\gamma}$, $\norm{\vh_{\gamma}}\le 1$ with bounded interaction degree $d$ (\autoref{sec:Ham}) and an operator $\vA$ supported on region $A\subset \Lambda$, let $\vH_{\ell}$ contain all terms $\vh_{\gamma}$ such that $dist(\gamma,A) < \ell-1$ for an integer $\ell$. Then,
\begin{align}
\norm{\vA_{\vH_\ell}(t)-\vA_{\vH}(t) } \lesssim \norm{\vA} \min\left(2,\labs{A}\frac{(2d\labs{t})^{\ell}}{\ell!}\right).
\end{align}    
\end{lem}

We will also need the following variant of the Lieb-Robinson bounds. 

\begin{lem}[Perturbing Hamiltonians]\label{lem:perturb_At}
Consider $\vF = \sum_{\gamma} \vf_{\gamma}$ and $\vF'=\sum_{\gamma} \vf'_{\gamma}$ with the same interaction graph of degree bounded by $d$ (as in~\autoref{sec:Ham}) and $\norm{\vf_{\gamma}},\norm{\vf'_{\gamma}}\le 1$. Then, for single site operator $\vA$, with $\norm{\vA} \le 1,$ 
\label{lem:perturb_far}
    \begin{align}
        \norm{\vA_{\vF'}(t)-\vA_{\vF}(t)} \lesssim \frac{1}{d}\sum_{\delta\in \Gamma} \norm{\vf_{\delta}-\vf'_{\delta}}\min\L( \frac{(2dt)^{(\mathrm{dist}(\delta,\vA)+1)}}{(\mathrm{dist}(\delta,\vA)+1)!},2 t \R).
    \end{align}
In particular, suppose $\vF' = \vF +\eta \vW$ such that all term in $\vW$ are far from $\vA$ with distance at least $\ell_0$,
\begin{align}
    \norm{\vA_{\vF'}(t)-\vA_{\vF}(t)} \lesssim \frac{\eta}{d}\sum_{\ell = \ell_0}^{\infty}S(\ell, \vA) \min\left( t,  \frac{(2dt)^{\ell+1}}{(\ell+1)!}\right),
\end{align}
where $S(\ell, \vA)$ is the number of terms distance $\ell$ from $\vA$. 
\end{lem}
\begin{rmk}
    The point of this variant is that the RHS scales linearly with $\eta$. Indeed, if we were to apply ~\autoref{lem:LRbounds} for both $\vF$ and $\vF',$ then the error from the RHS of~\autoref{lem:LRbounds} would not depend on $\eta.$
\end{rmk}
\begin{proof}
    We interpolate from $\vF = \sum_{\gamma} \vf_{\gamma}$ to $\vF'=\sum_{\gamma} \vf'_{\gamma}$ by changing one $\gamma$ at a time. With loss of generality, it suffices to consider one step $\vF = \sum_{\gamma} \vf_{\gamma}$ and $\vF' = \sum_{\gamma\ne \delta} \vf_{\gamma} + \vf'_{\delta}$ where $\norm{\vf_{\gamma}}, \norm{\vf'_{\delta}}\le 1$.
Then,
\begin{align}
    \norm{\e^{i\vF' t}\vA\e^{-i\vF' t} - \e^{i\vF t}\vA\e^{-i\vF t}}&=\lnorm{\e^{i(\vF +\Delta) t}\vA\e^{-i(\vF +\Delta) t}- \e^{i\vF t}\vA\e^{-i\vF t}} \\
    & \le \lnorm{\int_0^{t} [\Delta, \vA(t_1)](t-t_1) \rd t_1}\\
    &\le \int_0^{t} \norm{[\Delta, \vA(t_1)]}\rd t_1\\
    &\le \norm{\Delta} \int_0^{t} \frac{(2dt_1)^{\ell}}{\ell!}\rd t_1 \\
    &= \frac{\norm{\Delta}}{2d} \frac{(2dt)^{\ell+1}}{(\ell+1)!}.
\end{align}
The first equality sets $\Delta := \vh'_{\delta} - \vh_{\delta}$, second inequality uses Duhamel's identity for linear operators $e^{(\vC+\vD)t}-e^{\vC t} = \int_0^t e^{(\vC+\vD)(t-t_1)}\vD e^{\vC t_1} \rd t_1$ and the fourth inequality uses~\autoref{lem:LRbounds}, setting $\ell = \text{dist}(\delta,\vA)$.  In the case where the distance $\ell$ is too small, we also have the unconditional bound
\begin{align}
\norm{\e^{i\vF' t}\vA\e^{-i\vF' t} - \e^{i\vF t}\vA\e^{-i\vF t} } \le \int_0^{t} \norm{[\Delta, \vA(t_1)]}\rd t_1 \le 2t \|\Delta\|. 
\end{align}
Sum over all terms $\delta \in \Gamma$ to conclude the first claim.

To obtain the second claim, we organize the sum by the distance $\ell$
\begin{align}
\norm{\e^{i\vF' t}\vA\e^{-i\vF' t} - \e^{i\vF t}\vA\e^{-i\vF t}} &\le \frac{\eta}{d}\sum_{\ell = \ell_0}^{\infty}S(\ell,\vA) \min\left(2t,\frac{(2dt)^{\ell+1}}{(\ell+1)!}\right), 
\end{align}
using that the number of terms $\gamma$ that are distance $\ell$ from $\vA$ is bounded by $S(\ell,\vA)$, which concludes the proof.
\end{proof}
The rest of this section collects routine uses of Lieb-Robinson bounds.
\subsection{Proof of~\autoref{lem:high_freq}}
\label{sec:proof_high_freq}
\begin{proof}[Proof of~\autoref{lem:high_freq}] In this proof we suppress $\vA_{\vH'}(\omega') =: \vA(\omega')$.
Rewrite~\autoref{lem:bounds_imaginary_conjugation} to expose a decaying factor of $e^{-\beta \omega'}$
\begin{align}
    \hat{\vA}(\omega') = e^{-\beta_0\omega' +\beta_0^2\sigma^2} \cdot \widehat{(e^{\beta_0\vH'} \vA e^{-\beta_0\vH'})} (\omega'-2\sigma^2\beta_0)
\end{align}
and expand the RHS. For the imaginary time conjugation, we directly use a naive Taylor series that is suitable for small enough $\beta_0$
\begin{align}
    e^{\beta_0\vH'}\vA e^{-\beta_0\vH'} &= \sum_{k=0}^{\infty} \frac{1}{k!} \beta_0^k \CC_{\vH'}^k[\vA]\\
    &=\sum_{k=0}^{\infty} \sum_{\gamma_k \sim \cdots \sim \gamma_1\sim A}\frac{1}{k!} \beta_0^k  [\vh'_{\gamma_k},\cdots, [\vh'_{\gamma_1},\vA]\cdots].
\end{align}
Now, we study real-time evolution. For any operator $\vT_S$ (which will be the nested commutators) supported on a subset $S\subset \Lambda$ and normalized by $\norm{\vT_S}\le 1$, we introduce an annulus decomposition to exploit the locality of $\vG$ 
\begin{align}
    \norm{[e^{i\vH' t}\vT_S e^{-i\vH't},\vG]}
    &\le\sum_{\ell=\ell_0}^{\infty} \lnorm{[e^{i\vH_{\ell+1}' t}\vT_S e^{-i\vH_{\ell+1}'t} - e^{i\vH_{\ell}' t}\vT_S e^{-i\vH_{\ell}'t},\vG]} + \lnorm{[e^{i\vH_{\ell_0}' t}\vT_S e^{-i\vH_{\ell_0}'t} ,\vG]}\\
    &\lesssim \sum_{\ell=\ell_0}^{\infty}\min\left(\frac{(2d\labs{t})^{\ell}}{\ell!},1\right) \cdot V(\ell,S) + V(\ell_0-1,S) \le \labs{S} \L(\sum_{\ell=\ell_0}^{\infty}\frac{(2d\labs{t})^{\ell}}{\ell!} \cdot d^{\ell+2} +d^{\ell_0+1}\R),
\end{align}
using Lieb-Robinson bounds (\autoref{lem:LRbounds}), and that $[e^{i\vH_{\ell+1}' t}\vT_S e^{-i\vH_{\ell+1}'t} - e^{i\vH_{\ell}' t}\vT_S e^{-i\vH_{\ell}'t},\vG]$ is supported distance $\ell$ from set $S$, so the the number of $\vec{g}_{\gamma}$ that contributes to the commutator is bounded by a volume-bound $V(\ell).$ Here we can bound $V(\ell)\le \labs{S} d^{\ell+2}$. Thus, we obtain the bound

\begin{align} \norm{[e^{i\vH' t}\vT_S e^{-i\vH't},\vG]}\le \labs{S} \L(\sum_{\ell=\ell_0}^{\infty}\frac{(2d\labs{t})^{\ell}}{\ell!} \cdot d^{\ell+2} +d^{\ell_0+1}\R).
\end{align}

Express the time-domain expression for the operator Fourier transform,
\begin{align}
    \norm{[\vT_S(\omega'),\vG]} &\le \frac{1}{\sqrt{2\pi}}\int_{-\infty}^\infty \lnorm{[e^{i\vH' t}\vT_S e^{-i\vH't},\vG]} \labs{f(t)} \rd t\\
    &\lesssim \sqrt{\sigma}d^2\labs{S} \int_{-\infty}^{\infty} \L( \sum_{\ell=\ell_0}^{\infty}\frac{(2d^2\labs{t})^{\ell}}{\ell!} + d^{\ell_0-1}\R) e^{-\sigma^2 t^2} \rd t\\
    &\lesssim \frac{\labs{S}}{\sqrt{\sigma}}d^2 \L(d^{\ell_0-1} + \sum_{\ell=\ell_0}^{\infty}(\frac{2ed^2}{\sigma\sqrt{\ell}})^{\ell} \R)\\
    &\lesssim \frac{\labs{S}}{\sqrt{\sigma}}d^2\L(d^{\ell_0-1}+ 1/2 \R) \lesssim \frac{\labs{S}}{\sqrt{\sigma}} d^{\ell_0+1}, \quad \text{setting $\ell_0 =\lceil 16e^2d^4/\sigma^2\rceil$,}
\end{align}
where the third line uses that $\int_{-\infty}^{\infty} e^{-x^2} \labs{x}^{\ell}\rd x = \int_{0}^{\infty} e^{-y} y^{\ell/2-1}\rd y=\Gamma((\ell+1)/2)\le \lceil(\ell-1)/2\rceil!\le \ell^{\ell/2}$ and Stirling's approximation $1/\ell! \le (e/\ell)^{\ell}$ for integers $\ell \ge 1.$ The last line sums over a geometric series. Altogether, 
\begin{align}
    \lnorm{[\widehat{(e^{\beta_0\vH'} \vA e^{-\beta_0\vH'})} (\omega'), \vG]} &\le \lnorm{ \sum_{k=0}^{\infty} \sum_{\gamma_k,\cdots,\gamma_1}\frac{1}{k!} \beta_0^k  [[\vh'_{\gamma_k},\cdots, [\vh'_{\gamma_1},\vA]\cdots](\omega'),\vG]}\\
    &\le \|\vA\| \sum_{k=0}^{\infty} (2\beta_0d)^k dk \frac{d^{\ell_0+1}}{\sqrt{\sigma}} \lesssim \frac{d^{\ell_0+2}}{\sqrt{\sigma}}\|\vA\|, \tag*{setting $\beta_0 = 1/4d$,}
\end{align}
where the second lines plugs in $\vT_S \leftarrow [\vh'_{\gamma_k},\cdots, [\vh'_{\gamma_1},\vA]\cdots]/(2^{k} \norm{\vA})$, $|S| \leq dk$, that there are at most $k!d^k$ paths $\gamma_k \sim \cdots \sim \gamma_1\sim A$, and using
that $\labs{\sum_k kx^k} \le \frac{1}{(1-\labs{x})^2}$ for $\labs{x}\le 1.$ Integrating over $\omega'$ to obtain
\begin{align}
\labs{\int_{\labs{\omega'}\ge \Omega'} \braket{\vO,[\hat{\vA}(\omega'),\vG]}_{\vrho} \rd \omega'} &\le \int_{\labs{\omega'}\ge \Omega'} \norm{[\hat{\vA}(\omega'),\vG]} \|\vO\| \rd \omega'\\
&\lesssim \frac{d^{\ell_0+2}}{\sqrt{\sigma}}\int_{\omega' \ge \Omega'}\|\vO\|\|\vA\| e^{-\beta_0\omega' +\beta_0^2\sigma^2} \rd \omega'\\
&\lesssim \frac{d^{\ell_0+2}}{\beta_0\sqrt{\sigma}} e^{-\beta_0\Omega'+ \beta^2_0\sigma^2}\|\vO\|\|\vA\| \lesssim\frac{d^{4+16e^2d^4/\sigma^2}}{\sqrt{\sigma}} e^{-\Omega'/4d + \sigma^2/16d^2}\|\vO\|\|\vA\|,
\end{align}
as advertised.
\end{proof}

\subsection{Proof of~\autoref{lem:stability_Q}}
\label{sec:proof_stability}
\begin{proof}[Proof of~\autoref{lem:stability_Q}]

For (A), apply~\autoref{lem:LRbounds} for truncating the Hamiltonian $\vG\rightarrow \vG_{\ell}$ and then $\vK\rightarrow \vK_{\ell}$ 
\begin{align}
    &\labs{Q(\vO,\vG,\vA,\vK) - Q(\vO,\vG_{\ell},\vA,\vK)}\\ &\le \int_{-\infty}^{\infty}\int_{-\infty}^{\infty} (\labs{h_+(t')} +\labs{h_-(t')}) g_{\beta}(t) \cdot \lnorm{\vO^{\dagger}_{\vG} (t)-\vO^{\dagger}_{\vG_{\ell}}(t)} \rd t'\rd t\\
    &\le 2\int_{-\infty}^{\infty} (\labs{h_+(t')} +\labs{h_-(t')}) \rd t' \L(\int_{0\le t\le \ell/2e^2d} g_{\beta}(t) |\Supp(\vO)| \frac{(2dt)^{\ell}}{\ell!}\rd t + 2\int_{t> \ell/2e^2d} g_{\beta}(t) \rd t \R)\\
    &\lesssim \int_{-\infty}^{\infty} (\labs{h_+(t')} +\labs{h_-(t')}) \rd t' \L(\int_{t\le \ell/2e^2d} e^{-2\pi \labs{t}/\beta} e^{-\ell}\rd t/\beta + \int_{t> \ell/2e^2d} e^{-2\pi \labs{t}/\beta} \rd t/\beta \R)|\Supp(\vO)|\\
    &\lesssim \frac{e^{\beta \Omega'/2+\sigma^2\beta^2/4}}{\sqrt{\sigma}\beta} \L( e^{-\ell} + e^{-\pi\ell/e^2d\beta}\R)\label{eq:GGell}|\Supp(\vO)|,
\end{align}
where the second inequality bounds the late-time contribution by the trivial bound $\lnorm{\vO_{\vG}^{\dagger} (t)-\vO^{\dagger}_{\vG_{\ell}}(t)} \le 2$. The third inequality uses that $\frac{(2dt)^{\ell}}{\ell!} \le e^{-\ell}$ for $0 \le t \le \ell/2e^2d.$ Next, we change $\vK$ to $\vK_{\ell}$. The main difference is to split the integration range more carefully since the Heisenberg dynamics depends on $t+t'$
\begin{align}
    &\labs{Q(\vO,\vG_{\ell},\vA,\vK) - Q(\vO,\vG_{\ell},\vA,\vK_{\ell})}\\ 
    &\le \int_{-\infty}^{\infty}\int_{-\infty}^{\infty} (\labs{h_+(t')} +\labs{h_-(t')}) g_{\beta}(t) \cdot \lnorm{\vA_{\vK}(t+t')-\vA_{\vK_{\ell}}(t+t')} \rd t'\rd t\\
    &\lesssim \int_{|t'|> \ell/4e^2d}(\labs{h_+(t')} +\labs{h_-(t')} )\int_{-\infty}^{\infty}g_{\beta}(t)\rd t\rd t'\\
    &+ \int_{|t'|\le \ell/4e^2d}(\labs{h_+(t')} +\labs{h_-(t')}) \L(\int_{|t|\le \ell/4e^2d} e^{-2\pi \labs{t}/\beta} |\Supp(\vA)| \frac{(2d|t+t'|)^\ell}{\ell!}  \rd t/\beta + \int_{|t|> \ell/4e^2d} e^{-2\pi \labs{t}/\beta} \rd t/\beta \R)  \rd t'\\
     &\lesssim \frac{e^{\beta \Omega'/2+\sigma^2\beta^2/4} }{\beta\sqrt{\sigma}}\L(e^{-\sigma^2\ell^2/16e^4d^2}+ e^{-\ell} + e^{-\pi\ell/2e^2d\beta}\R)|\Supp(\vA)|.\label{eq:KKell}
\end{align}

For the case (B) with extensive perturbation $\vK\rightarrow \vK',\vG\rightarrow\vG',$ the expression is very much the same as~\eqref{eq:GGell},\eqref{eq:KKell}, except we used \autoref{lem:perturb_At} instead of \autoref{lem:LRbounds}.
\begin{align}
    &\labs{Q(\vO,\vG,\vA,\vK) - Q(\vO,\vG',\vA,\vK)}\\
    &\le \int_{-\infty}^{\infty}\int_{-\infty}^{\infty} (\labs{h_+(t')} +\labs{h_-(t')}) g_{\beta}(t) \cdot \lnorm{\vO^{\dagger}_{\vG} (t)-\vO^{\dagger}_{\vG'}(t)} \rd t'\rd t\\
    &\lesssim 2\kappa/d \sum_{\ell=\ell_0}^{\infty} S(\ell,\vO)\int_{-\infty}^{\infty} (\labs{h_+(t')} +\labs{h_-(t')}) \rd t' \L(\int_{0\le t\le (\ell+1)/2e^2d} g_{\beta}(t) \frac{(2dt)^{\ell+1}}{(\ell+1)!}\rd t + \int_{t> (\ell+1)/2e^2d} t g_{\beta}(t) \rd t \R)\\
    &\lesssim \kappa/d \sum_{\ell=\ell_0}^{\infty} S(\ell,\vO) \frac{e^{\beta \Omega'/2+\sigma^2\beta^2/4}}{\sqrt{\sigma}\beta} \L( e^{-\ell-1} + \frac{(\ell+1) \beta}{d}e^{-\pi(\ell+1)/e^2d\beta}\R).
\end{align}
And,
\begin{align}
    &\labs{Q(\vO,\vG',\vA,\vK) - Q(\vO,\vG',\vA,\vK')} \\
    &\le \int_{-\infty}^{\infty}\int_{-\infty}^{\infty} (\labs{h_+(t')} +\labs{h_-(t')}) g_{\beta}(t) \cdot \lnorm{\vA_{\vK}(t+t')-\vA_{\vK'}(t+t')} \rd t'\rd t\\
    &\lesssim \kappa \sum_{\ell=\ell_0}^{\infty} S(\ell, \vA)\int_{|t'|> (\ell+1)/4e^2d}(\labs{h_+(t')} +\labs{h_-(t')} \int_{-\infty}^{\infty}\labs{t+t'}g_{\beta}(t)\rd t\rd t'\\
    &+ \kappa \sum_{\ell=\ell_0}^{\infty} S(\ell,\vA)\int_{|t'|\le (\ell+1)/4e^2d}(\labs{h_+(t')} +\labs{h_-(t')}) \L(\int_{|t|\le (\ell+1)/4e^2d} e^{-2\pi \labs{t}/\beta} e^{-\ell-1}\rd t/\beta \right. \\
    &\hspace{10cm} \left. 
 +\int_{|t|> (\ell+1)/4e^2d} e^{-2\pi \labs{t}/\beta} \labs{t+t'}\rd t/\beta\rd t'  \R)
    \\
    &\lesssim \kappa \sum_{\ell=\ell_0}^{\infty} S(\ell,\vA) \frac{e^{\beta \Omega'/2+\sigma^2\beta^2/4} }{\beta\sqrt{\sigma}}\L((\beta+\frac{\ell}{4ed}) e^{-\sigma^2(\ell+1)^2/16e^4d^2}+ e^{-\ell-1} + \frac{(\ell+1)\beta)}{d}e^{-\pi(\ell+1)/2e^2d\beta}\R).
\end{align}
Simplify the bounds by $\ell+1\rightarrow \ell$ to obtain the advertised result.

The proof of case (C) similarly makes use of~\autoref{lem:perturb_At}.
\begin{align}
 &\labs{Q(\vO,\vG,\vA,\vK) - Q(\vO,\vG',\vA,\vK)}\\
    &\le \int_{-\infty}^{\infty}\int_{-\infty}^{\infty} (\labs{h_+(t')} +\labs{h_-(t')}) g_{\beta}(t) \cdot \lnorm{\vO^{\dagger}_{\vG} (t)-\vO^{\dagger}_{\vG'}(t)} \rd t'\rd t\\
    &\lesssim \frac{\kappa}{d} V(\ell_0, \vO) \int_{-\infty}^{\infty} (\labs{h_+(t')} +\labs{h_-(t')}) \rd t' \L( \int_{-\infty}^{\infty} |t| g_{\beta}(t) \rd t \R)\\
    &\lesssim \frac{\kappa \beta}{d} V(\ell_0, \vO) \frac{e^{\beta \Omega'/2+\sigma^2\beta^2/4}}{\sqrt{\sigma}\beta} .   
\end{align}
And
\begin{align}
    &\labs{Q(\vO,\vG',\vA,\vK) - Q(\vO,\vG',\vA,\vK')}\\
    &\le \int_{-\infty}^{\infty}\int_{-\infty}^{\infty} (\labs{h_+(t')} +\labs{h_-(t')}) g_{\beta}(t) \cdot \lnorm{\vA_{\vK} (t+t')-\vA_{\vK'}(t+t')} \rd t'\rd t\\
    &\lesssim \frac{\kappa}{d} V(\ell_0,\vA) \int_{-\infty}^{\infty}\int_{-\infty}^{\infty} (\labs{h_+(t')} +\labs{h_-(t')}) g_{\beta}(t) \cdot  |t+t'| \rd t'\rd t\\
    &\lesssim \frac{\kappa \beta}{d} V(\ell_0,\vA) \frac{e^{\beta \Omega'/2+\sigma^2\beta^2/4}}{\sqrt{\sigma}\beta}.
\end{align}
\end{proof}

\end{document}